\documentclass[onefignum,onetabnum]{siamart251216}

\usepackage{lipsum}
\usepackage{amsfonts}
\usepackage{graphicx}
\usepackage{epstopdf}
\usepackage{algorithmic}
\usepackage{subcaption}
\usepackage{cleveref}
\ifpdf
  \DeclareGraphicsExtensions{.eps,.pdf,.png,.jpg}
\else
  \DeclareGraphicsExtensions{.eps}
\fi

\usepackage{enumitem}
\setlist[enumerate]{leftmargin=.5in}
\setlist[itemize]{leftmargin=.5in}

\newsiamremark{remark}{Remark}
\newsiamremark{hypothesis}{Hypothesis}
\crefname{hypothesis}{Hypothesis}{Hypotheses}
\newsiamthm{claim}{Claim}
\newsiamremark{fact}{Fact}
\crefname{fact}{Fact}{Facts}
\Crefname{ALC@unique}{Line}{Lines}

\newcommand{\bszero}{\boldsymbol{0}}
\newcommand{\bsh}{\boldsymbol{h}}

\newcommand{\bsk}{\boldsymbol{k}}

\newcommand{\bsv}{\boldsymbol{v}}
\newcommand{\bsx}{\boldsymbol{x}}
\newcommand{\bsy}{\boldsymbol{y}}
\newcommand{\bsz}{\boldsymbol{z}}
\newcommand{\bsalpha}{\boldsymbol{\alpha}}

\newcommand{\bsDelta}{\boldsymbol{\Delta}}

\newcommand{\NN}{\mathbb{N}}

\newcommand{\QQ}{\mathbb{Q}}
\newcommand{\RR}{\mathbb{R}}

\newcommand{\ZZ}{\mathbb{Z}}

\newcommand{\Ccal}{\mathcal{C}}

\newcommand{\Scal}{\mathcal{S}}

\DeclareMathOperator{\Span}{span}
\DeclareMathOperator{\diam}{diam}
\DeclareMathOperator{\vol}{vol}

\allowdisplaybreaks

\headers{Space-filling lattice designs}{N. Sakai and T. Goda}

\title{Space-filling lattice designs for\\ computer experiments\thanks{Submitted to the editors DATE.}}

\author{Naoki Sakai\thanks{Department of Energy Science and Engineering, Stanford University, 537 Panama Mall, Stanford, CA 94305, USA 
  (\email{naoki65@stanford.edu}).}
\and Takashi Goda\thanks{Graduate School of Engineering, The University of Tokyo, 7-3-1 Hongo, Bunkyo-ku, Tokyo 113-8656, Japan 
  (\email{goda@frcer.t.u-tokyo.ac.jp}).}
}

\usepackage{amsopn}

\ifpdf
\hypersetup{
  pdftitle={Space-filling lattice designs for computer experiments},
  pdfauthor={N. Sakai and T. Goda}
}
\fi

\begin{document}

\maketitle

\begin{abstract}
    This paper investigates the construction of space-filling designs for computer experiments. The space-filling property is characterized by the covering and separation radii of a design, which are integrated through the unified criterion of quasi-uniformity. We focus on a special class of designs, known as quasi-Monte Carlo (QMC) lattice point sets, and propose two construction algorithms. The first algorithm generates rank-1 lattice point sets as an approximation of quasi-uniform Kronecker sequences, where the generating vector is determined explicitly. As a byproduct of our analysis, we prove that this explicit point set achieves an isotropic discrepancy of $O(N^{-1/d})$. The second algorithm utilizes Korobov lattice point sets, employing the Lenstra--Lenstra--Lov\'{a}sz (LLL) basis reduction algorithm to identify the generating vector that ensures quasi-uniformity. Numerical experiments are provided to validate our theoretical claims regarding quasi-uniformity. Furthermore, we conduct empirical comparisons between various QMC point sets in the context of Gaussian process regression, showcasing the efficacy of the proposed designs for computer experiments.
\end{abstract}



\sloppy

\section{Introduction}\label{sec:introduction}

In the field of computer experiments, maximin and minimax distance-based designs have gained prominence due to their superior space-filling properties \cite{FLS06,HRDH11,JMY90,J16,P17,SWN03}. Conceptually, maximin designs prioritize the separation between sampling points to ensure that they are well-dispersed across the domain, while minimax designs aim to cover the target space as efficiently as possible by minimizing the largest unexplored regions. These two geometric objectives, i.e., separation and coverage, are naturally unified by the concept of quasi-uniformity, which aims to attain both properties simultaneously and has become one of the central criteria for assessing the quality of experimental designs recently \cite{PM12,PZ23}.

The utility of quasi-uniform designs extends across various numerical frameworks, including radial basis function approximation, kernel interpolation, and Gaussian process regression \cite{Sch95,SW06,T20,W05,WSH21,WBG21}. To appreciate the necessity of quasi-uniformity, consider kernel interpolation with standard Sobolev kernels as a primary example. In this context, the geometric properties of the point set dictate both the approximation accuracy and the numerical stability. When the target function resides within the associated native reproducing kernel Hilbert space, the approximation error is typically bounded by a power of the covering radius. However, minimizing the covering radius is not the sole objective. The separation radius is equally critical, as it governs the condition number of the kernel matrix (Gram matrix); a shrinking separation radius leads to severe ill-conditioning, compromising numerical stability \cite{Sch95}. Furthermore, in scenarios where the target function possesses lower regularity than the kernel implies, i.e., when approximating a ``rough'' function that lies outside the native space, theoretical error bounds often explicitly depend on the mesh ratio (the ratio of the covering radius to the separation radius), rather than solely on the covering radius \cite{SW06,T20,WBG21}. Consequently, constructing quasi-uniform designs for computer experiments is critical for achieving a balance between high-fidelity convergence and numerical stability. Furthermore, in the context of Bayesian optimization, which frequently employs Gaussian process regression as a surrogate model, the use of quasi-uniform designs in the initial sampling stage has empirically been shown to accelerate the convergence of the optimization process \cite{KIWINTSG25}.

Quasi-Monte Carlo (QMC) point sets and sequences, also known as low-discrepancy point sets and sequences, are widely recognized for their ability to achieve faster convergence rates in numerical integration over the unit hypercube $[0,1]^d$ than random sampling \cite{DKS13,KN74,N92}. Given their excellent equi-distribution properties, they have also been considered prime candidates for space-filling quasi-uniform designs, see, for instance, \cite[Section~3.1]{WBG21}. However, the specific geometric property of quasi-uniformity in the QMC context has remained relatively under-explored until very recently. Recent studies have begun to fill this gap, revealing highly non-trivial results: many widely used QMC sequences---including the two-dimensional Sobol', Faure, and Halton sequences---fail to satisfy the requirements for quasi-uniformity (see \cite{G24a}, \cite{DGS25}, and \cite{GHS25}, respectively). More precisely, the separation radius of these sequences decays at a faster rate than $N^{-1/d}$, violating the conditions necessary for quasi-uniformity. In contrast, a recent work \cite{DGLPS25}, together with a prior work \cite{G24b}, proves that lattice point sets and Kronecker sequences can achieve quasi-uniformity, provided that their underlying parameters are appropriately selected.

While \cite{DGLPS25} provided an explicit construction for quasi-uniform Kronecker sequences, the results for lattice point sets remained purely existential; neither an explicit construction nor a constructive algorithm was previously known. Addressing this gap is the primary objective of the present work. In this paper, we propose two distinct methodologies for constructing space-filling, quasi-uniform lattice designs for computer experiments as follows:

\begin{enumerate}
    \item \emph{Explicit rank-1 lattice construction}: We provide an explicit rank-1 lattice point set construction as a discrete approximation of the quasi-uniform Kronecker sequences introduced in \cite{DGLPS25}. A characteristic of this method is that the number of sampling points $N$ is not arbitrarily chosen but is determined by the Diophantine properties of the parameters defining the underlying Kronecker sequence. Despite this constraint, the method provides the parameters, i.e., the generating vectors, for rank-1 lattice point sets explicitly. Furthermore, this result yields rank-1 lattice point sets with low isotropic discrepancy, thereby providing a partial solution to an open problem regarding the explicit construction of such sets \cite[Corollary 5.27]{DKP22}.
    \item \emph{Algorithmic Korobov lattice construction}: Focusing on a class of Korobov lattice point sets, we develop a constructive algorithm that leverages the Lenstra–Lenstra–Lovász (LLL) lattice basis reduction algorithm \cite{LLL82} to efficiently identify the generating vector that ensures quasi-uniformity. This algorithm is applicable for any prime number $N$, and its computational cost is highly efficient, growing at most polynomially in the dimension $d$ and only quasi-linearly in $N$.
\end{enumerate}

The rest of this paper is organized as follows. In \cref{sec:preliminaries}, we introduce the necessary notation and formal definitions, including the criteria for quasi-uniformity and the construction of QMC lattice point sets. We also introduce the concept of successive minima from the geometry of numbers \cite{C97,S89}, and establish their connection to the quasi-uniformity of lattice designs. In \cref{sec:construction}, we detail our two proposals for constructing space-filling lattice designs: explicit rank-1 lattice construction and algorithmic Korobov lattice construction. We provide theoretical results demonstrating that both approaches yield quasi-uniform designs. Finally, \cref{sec:experiments} presents numerical results and performance evaluations in the context of Gaussian process regression, followed by concluding remarks in \cref{sec:conclusion}.

\section{Preliminaries}\label{sec:preliminaries}

Throughout this paper, let $\ZZ$ be the set of integers, and $\NN$ be the set of positive integers. We also use the standard asymptotic notation: for two non-negative functions $f,g$, we write $f(N)=O(g(N))$ (resp. $f(N)=\Omega(g(N))$) if $f\le Cg$ (resp. $f\ge Cg$) for some constant $C>0$ and sufficiently large $N$, and we write $f(N)=\Theta(g(N))$ if both conditions hold.

We focus on the $d$-dimensional unit cube $[0,1]^d$ as our target domain. This choice is without loss of generality: while practical applications often involve general rectangular domains, any such domain can be mapped from $[0,1]^d$ via a linear transformation. Since such transformations only alter the covering and separation radii by a constant factor determined by the domain's side lengths, the asymptotic properties and the quasi-uniformity of the designs are preserved.

\subsection{Quasi-uniform designs}\label{subsec:qu_designs}

Let $P_N \subset [0,1]^d$ be a set of $N$ points. To evaluate how well these points fill the space, we consider the covering radius and the separation radius. The \emph{covering radius}, or fill distance, of $P_N$ with respect to $[0,1]^d$ is defined as
\[ h(P_N) = \sup_{\bsx \in [0,1]^d} \min_{\bsy \in P_N} \|\bsx - \bsy\|_2, \]
which measures the largest distance from any point in the domain to the nearest point in $P_N$. The \emph{separation radius} of $P_N$ is defined as
\[ q(P_N) = \frac{1}{2}\min_{\substack{\bsx,\bsy\in P_N\\ \bsx\neq \bsy}} \|\bsx - \bsy\|_2, \]
representing half the minimum distance between any two distinct points in the set. 

These two quantities admit the following geometric interpretations. If we consider Euclidean balls of radius $r$ centered at each point in $P_N$, the covering radius $h(P_N)$ is the minimum radius such that the union of the closed balls covers the entire domain $[0,1]^d$. Conversely, the separation radius $q(P_N)$ is the maximum radius such that the interiors of the balls are pairwise disjoint. From these interpretations, it follows that for any set of $N$ points $P_N$, there exist positive constants $C_{1,d},C_{2,d}$, depending only on the dimension $d$, such that
\begin{align}\label{eq:bounds_h_and_q}
    h(P_N)\ge C_{1,d} N^{-1/d} \quad \text{and}\quad q(P_N)\le C_{2,d} N^{-1/d},
\end{align}
respectively (see, for instance, \cite{DGLPS25,DGS25,PZ23}).

Let $\Scal\subset [0,1]^d$ be an infinite sequence of points, and let $P_N$ denote the set consisting of the first $N$ points of $\Scal$. We formally define the quasi-uniformity of the sequence $\Scal$ as follows.

\begin{definition}
    An infinite sequence of points $\Scal$ is \emph{quasi-uniform} over $[0,1]^d$ if there exists a positive constant $C_d$, independent of $N$, such that
    \begin{align}\label{eq:quasi-uniform_condition}
        \rho(P_N) := \frac{h(P_N)}{q(P_N)}\le C_d,
    \end{align}
    for any $N\ge 2$. The quantity $\rho(P_N)$ is referred to as the \emph{mesh ratio} of $P_N$ with respect to $[0,1]^d$.
\end{definition}

As a consequence of the bounds \eqref{eq:bounds_h_and_q}, for a quasi-uniform sequence, both the covering and separation radii must satisfy $h(P_N)=\Theta(N^{-1/d})$ and $q(P_N)=\Theta(N^{-1/d})$. This implies that a quasi-uniform sequence fills the space at the optimal rate while maintaining a minimum separation between points, also at the optimal rate.

We also define quasi-uniformity for a sequence of (not necessarily extensible) point sets. It was shown in \cite{DGLPS25,DGS25} that if the condition \eqref{eq:quasi-uniform_condition} is satisfied along a subsequence of sample sizes $N_1<N_2<\cdots$ such that $N_{i+1}\le cN_i$ for some constant $c>1$, the underlying space-filling property is essentially preserved for all $N$. Motivated by this observation, we adopt the following definition:

\begin{definition}\label{def:quasi-uniform_point_set}
    Let $(P_{N_i})_{i\in \NN}$ be a set of point sets such that $P_{N_i}\subset [0,1]^d$, where $N_i=|P_{N_i}|$ denotes the number of points in the $i$-th set, and $N_i<N_{i+1}\le cN_i$ for a constant $c>1$. If there exists a positive constant $C_d$, independent of $i$, such that
    \[  \rho(P_{N_i})\le C_d  \]
    holds for all $i\in \NN$, we say that $(P_{N_i})_{i\in \NN}$ is a quasi-uniform family of point sets.
\end{definition}

\subsection{Lattices}\label{subsec:lattice}

A \emph{lattice} $\Lambda \subset \RR^d$ is a discrete additive subgroup of $\RR^d$. A full rank lattice is a lattice $\Lambda$ that can be represented as the set of all integer linear combinations of $d$ linearly independent vectors:
\[ \Lambda = T\ZZ^d = \left\{ T\bsk \mid \bsk \in \ZZ^d \right\}, \]
where $T \in \RR^{d \times d}$ is an invertible matrix called the \emph{generating matrix} (or \emph{basis matrix}). The columns of $T$ form a basis of the lattice. We define the covering and separation radii of $\Lambda$ with respect to the Euclidean space $\RR^d$ by
\[ h(\Lambda) := \sup_{\bsx \in \RR^d} \min_{\bsy \in \Lambda} \|\bsx - \bsy\|_2, \quad \text{and} \quad q(\Lambda) := \frac{1}{2}\min_{\substack{\bsx,\bsy\in \Lambda\\ \bsx\neq \bsy}} \|\bsx - \bsy\|_2 = \frac{1}{2} \min_{\bsx \in \Lambda \setminus \{\bszero\}} \|\bsx\|_2, \]
respectively. Likewise, the mesh ratio is defined by $\rho(\Lambda) := h(\Lambda) / q(\Lambda)$. The separation radius $q(\Lambda)$ is equivalent to half the length of the shortest non-zero vector in the lattice.

For a given number of points $N$ and a generating vector $\bsz = (z_1, \dots, z_d) \in \{1, \dots, N-1\}^d$, let us consider the set 
\[ P_{N,\bsz} := \left\{ \bsx_n = \left( \left\{ \frac{nz_1}{N} \right\}, \dots, \left\{ \frac{nz_d}{N} \right\} \right) \mid n=0, 1, \dots, N-1 \right\} \subset [0,1]^d, \]
which is called a \emph{rank-1 lattice point set}, where $\{x\} = x - \lfloor x \rfloor$ denotes the fractional part of $x$. This point set has been extensively studied in the context of QMC integration \cite{DKP22, N92, SJ94}. Throughout this paper, we assume $\gcd(N, z_1, \dots, z_d) = 1$, which ensures that $P_{N,\bsz}$ consists of $N$ distinct points.

Such a point set can be viewed as the intersection $P_{N,\bsz} = \Lambda \cap [0,1)^d$ with an infinite lattice $\Lambda = \frac{1}{N}\bsz \ZZ + \ZZ^d$. Since $\ZZ^d\subseteq \Lambda$ and the unit cube $[0,1)^d$ contains exactly $N$ distinct points of $\Lambda$ under the condition $\gcd(N, z_1, \dots, z_d) = 1$, the volume of the fundamental domain of $\Lambda$ is $1/N$. Therefore, there always exists a generating matrix $T \in \RR^{d \times d}$ for $\Lambda$ such that $\det(T) = 1/N$. For instance, when $\gcd(N, z_1) = 1$, a generating matrix $T$ can be explicitly constructed as
\begin{align}\label{eq:generating_matrix}
T = \begin{pmatrix} 1/N & 0 & \dots & 0 \\ z_1^{-1}z_2/N & 1 & \dots & 0 \\ \vdots & \vdots & \ddots & \vdots \\ z_1^{-1}z_d/N & 0 & \dots & 1 \end{pmatrix},
\end{align}
where $z_1^{-1} \in \ZZ_N$ is the multiplicative inverse of $z_1$ modulo $N$. More generally, for any $\bsz$ with at least one component coprime to $N$, we can achieve $\gcd(N, z_1) = 1$ by reordering coordinates, allowing us to use this explicit form. A well-known special case is the \emph{Korobov lattice point set}, where 
\[ \bsz = (1, a, a^2, \dots, a^{d-1}) \pmod N\]
for some integer $a\in \ZZ_N$, which always satisfies $\gcd(N, z_1) = 1$.

\begin{remark}
It should be noted that the term ``rank-1'' refers to the fact that a single vector $\bsz$ generates $P_{N,\bsz}$ as a finite abelian group (specifically, $P_{N,\bsz} \cong \ZZ_N$ when $\gcd(N, z_1, \dots, z_d)=1$). This is distinct from the rank of the underlying infinite lattice $\Lambda \subset \RR^d$, which is $d$ (full rank) since its generating matrix $T$ is invertible. 
\end{remark}

The dual lattice of $\Lambda$, which plays an important role in the subsequent analysis, is defined by
\[ \Lambda^{\perp} := \left\{ \bsy \in \RR^d \mid \bsx \cdot \bsy \in \ZZ \quad \text{for all } \bsx \in \Lambda \right\}. \]
Let $T$ be a generating matrix for $\Lambda$. Then, it is easily seen that the dual lattice $\Lambda^{\perp}$ itself is also a lattice and that its generating matrix is $(T^{-1})^{\top}$. In particular, if $T$ is given by \eqref{eq:generating_matrix}, the generating matrix for $\Lambda^{\perp}$ is
\begin{align}\label{eq:dual_matrix}
(T^{-1})^{\top} = \begin{pmatrix} N & -z_1^{-1}z_2 & \dots & -z_1^{-1}z_d \\ 0 & 1 & \dots & 0 \\ \vdots & \vdots & \ddots & \vdots \\ 0 & 0 & \dots & 1 \end{pmatrix} \pmod{N},
\end{align}
for which $\det((T^{-1})^{\top}) = N$ holds.

\subsection{Successive minima and quasi-uniformity}\label{subsec:successive_minima}

To further characterize the geometry of a lattice $\Lambda$, we introduce the concept of \emph{successive minima} from the geometry of numbers \cite{C97, S89}. For $j=1,\ldots,d$, the $j$-th successive minimum of $\Lambda$ with respect to the Euclidean ball $B_d:=\left\{ \bsx\in \RR^d \mid \|\bsx\|_2\le 1 \right\}$ is defined as
\[ \lambda_j(\Lambda) := \inf\left\{ r>0 \mid \dim(\Span(\Lambda\cap rB_d) )\ge j \right\}. \]
By definition, $\lambda_j(\Lambda)$ represents the radius of the smallest ball centered at the origin that contains $j$ linearly independent lattice vectors.

The successive minima are closely related to the covering and separation radii of the lattice. From the definition of $q(\Lambda)$ in \cref{subsec:lattice}, it immediately follows that
\[ q(\Lambda) = \frac{\lambda_1(\Lambda)}{2}. \]
Regarding the covering radius of the lattice, it is known that
\[ \frac{1}{2}\lambda_d(\Lambda) \le h(\Lambda) \le \frac{\sqrt{d}}{2}\lambda_d(\Lambda), \]
see, for instance, \cite[Theorem~7.9]{MG02}. Moreover, in \cite{B93}, Banaszczyk proved the following transference theorems:
\begin{align}\label{eq:transference}
    \frac{1}{2} \le h(\Lambda) \lambda_1(\Lambda^{\perp}) \le \frac{d}{2} \quad \text{and} \quad 1 \le \lambda_j(\Lambda)\lambda_{d-j+1}(\Lambda^{\perp}) \le d, 
\end{align}
for all $j=1,\ldots,d$. These bounds imply that the mesh ratio $\rho(\Lambda)$ is fundamentally characterized by the balance of successive minima:
\begin{align}\label{eq:bound_mesh_raio}
    \max\left( \frac{\lambda_d(\Lambda)}{\lambda_1(\Lambda)}, d^{-1} \frac{\lambda_d(\Lambda^{\perp})}{\lambda_1(\Lambda^{\perp})}\right)\le \rho(\Lambda) \le \min\left( \sqrt{d}\frac{\lambda_d(\Lambda)}{\lambda_1(\Lambda)}, d\frac{\lambda_d(\Lambda^{\perp})}{\lambda_1(\Lambda^{\perp})}\right).
\end{align}
This sandwiching of $\rho(\Lambda)$ shows that the mesh ratio is small if and only if the successive minima of the lattice (or its dual lattice) are well-balanced.

Applying Minkowski's second theorem, we obtain the following result:

\begin{lemma}\label{lem:rho_inf_bound}
    Let $\Lambda \subset \RR^d$ be a lattice with a generating matrix $T$, and $\Lambda^{\perp}$ be its dual lattice. Then,
    \[ \rho(\Lambda) \le C_d \min\left( \sqrt{d}\frac{|\det(T)|}{(\lambda_1(\Lambda))^d}, d\frac{|\det(T^{-1})|}{(\lambda_1(\Lambda^{\perp}))^d}\right), \]
    with $C_d=2^d\Gamma(1+d/2)/\pi^{d/2}$ where $\Gamma$ denotes the gamma function.
\end{lemma}

\begin{proof}
    Minkowski's second theorem states that
    \[ \prod_{j=1}^{d}\lambda_j (\Lambda) \le \frac{2^d\Gamma(1+d/2)}{\pi^{d/2}}\left| \det(T)\right|, \]
    see \cite{C97,S89}. Since the successive minima are ordered, i.e., $\lambda_1(\Lambda) \le \lambda_2(\Lambda) \le \cdots \le \lambda_d(\Lambda)$, we obtain
    \[ \lambda_d(\Lambda) \le \frac{2^d\Gamma(1+d/2)}{\pi^{d/2}}\left| \det(T)\right| \prod_{j=1}^{d-1}\frac{1}{\lambda_j (\Lambda)} \le \frac{2^d\Gamma(1+d/2)}{\pi^{d/2}}\frac{\left| \det(T)\right|}{(\lambda_1 (\Lambda))^{d-1}}. \]    
    Similarly, for the dual lattice $\Lambda^{\perp}$, we have
    \[ \lambda_d(\Lambda^{\perp}) \le \frac{2^d\Gamma(1+d/2)}{\pi^{d/2}}\left| \det(T^{-1})\right| \prod_{j=1}^{d-1}\frac{1}{\lambda_j (\Lambda^{\perp})} \le \frac{2^d\Gamma(1+d/2)}{\pi^{d/2}}\frac{\left| \det(T^{-1})\right|}{(\lambda_1 (\Lambda^{\perp}))^{d-1}}. \]    
    By applying these estimates to the upper bound in \eqref{eq:bound_mesh_raio}, we obtain the desired inequality.
\end{proof}

\begin{remark}
    The reciprocal of the length of the shortest non-zero vector in the dual lattice, $1/\lambda_1(\Lambda^{\perp})$, is called the \emph{spectral test}, which is a standard quality measure for linear congruential pseudo-random number generators \cite{LE88,LE99}. In that context, $1/\lambda_1(\Lambda^{\perp})$ represents the maximum distance between adjacent parallel hyperplanes that cover all lattice points. A larger $\lambda_1(\Lambda^{\perp})$ corresponds to a smaller distance between these hyperplanes, indicating a more ``uniform'' distribution of points. The above lemma provides a theoretical bridge between this classical metric and the space-filling property.
\end{remark}

Before moving on to the next section, we show that the mesh ratio of the point set $P=\Lambda \cap [0,1)^d$ with respect to $[0,1]^d$ is bounded by the mesh ratio of the lattice $\Lambda$ with respect to $\RR^d$, up to a constant factor. Although a similar result was proven in \cite{DGLPS25}, our argument below directly deals with the Euclidean norm, instead of the $\ell_{\infty}$-norm as considered in \cite{DGLPS25}.

\begin{lemma}\label{lem:mesh_ratio_comparison}
    Let $\Lambda\subset \RR^d$ be a lattice and $P=\Lambda\cap [0,1)^d$, Assume $|P|\ge 2$. The separation and covering radii of $P$ with respect to $[0,1]^d$ satisfy
    \[ q(P)\ge q(\Lambda)\quad \text{and}\quad h(P)\le \begin{cases} 2\sqrt{d}\, h(\Lambda) & \text{if $h(\Lambda)\ge 1/2$,}\\ (1+\sqrt{d})h(\Lambda) & \text{otherwise.} \end{cases}\]
    Consequently, the mesh ratio of $P$ is bounded by
    \[ \rho(P)\le \begin{cases} 2\sqrt{d}\, \rho(\Lambda) & \text{if $h(\Lambda)\ge 1/2$,}\\ (1+\sqrt{d})\rho(\Lambda) & \text{otherwise.} \end{cases}.\]
\end{lemma}

\begin{proof}
    The inequality $q(P)\ge q(\Lambda)$ follows immediately from the fact that $P\subset \Lambda$.

    For the covering radius, let us write $\epsilon := h(\Lambda)$. If $\epsilon\ge 1/2$, the result is clear as we have
    \[ h(P)\le \diam([0,1]^d) = \sqrt{d}\le 2\sqrt{d}\, \epsilon.\]
    
    Assume $\epsilon<1/2$. 
    For any $\bsx \in [0,1]^d$, we can choose a ``shifted center'' $\bsz \in [\epsilon, 1-\epsilon]^d$ such that the Euclidean ball $B(\bsz,\epsilon)=\{\bsy\in \RR^d\mid \|\bsy-\bsz\|_2\le \epsilon\}$ is contained in $[0,1]^d$. Specifically, we set
    \[ z_j = \begin{cases}
        \epsilon & \text{if $x_j<\epsilon$,}\\
        1-\epsilon & \text{if $x_j>1-\epsilon$,}\\
        x_j & \text{otherwise,}
    \end{cases}\]
    for each $j=1,\ldots,d$. By the definition of $h(\Lambda)$, there exist a lattice point $\bsy^*\in \Lambda$ such that $\|\bsy^*-\bsz\|_2\le \epsilon$. Since we have $B(\bsz,\epsilon)\subset [0,1]^d$, it follows that $\bsy^*\in P$. Now, we evaluate the distance between $\bsx$ and $\bsy^*$. By the triangle inequality, it holds that
    \[ \|\bsx-\bsy^*\|_2\le \|\bsx-\bsz\|_2+\|\bsz-\bsy^*\|_2\le \|\bsx-\bsz\|_2+\epsilon.\]
    From the definition of $\bsz$, we have $|x_j-z_j|\le \epsilon$ for all $j$, which implies $\|\bsx-\bsz\|_2\le \sqrt{d}\epsilon$. Thus, $\|\bsx-\bsy^*\|_2\le (1+\sqrt{d})\epsilon$. Since this holds for any $\bsx\in [0,1]^d$, we obtain $h(P)\le (1+\sqrt{d})h(\Lambda)$.
    
    The bound on the mesh ratio for $P$ follows immediately.
\end{proof}

This lemma guarantees that finding a lattice $\Lambda$ with a small mesh ratio $\rho(\Lambda)$ is sufficient to construct a quasi-uniform point set $P$. Combining this with \cref{lem:rho_inf_bound}, our goal reduces to finding a generating vector $\bsz$ that maximizes $\lambda_1(\Lambda)$ or $\lambda_1(\Lambda^\perp)$.

\begin{remark}
    For any fixed shift $\bsDelta\in \RR^d$, consider the shifted lattice $\Lambda_{\bsDelta}:=\Lambda+\bsDelta$ and the corresponding finite point set $P_{\bsDelta}=\Lambda_{\bsDelta}\cap [0,1)^d$. Since the covering and separation radii of a lattice are invariant under translation, we have $h(\Lambda_{\bsDelta})=h(\Lambda)$ and $q(\Lambda_{\bsDelta})=q(\Lambda)$. Thus, provided $|P_{\bsDelta}|\ge 2$, the bounds in the above lemma remain valid for the shifted point set $P_{\bsDelta}$:
    \[ \rho(P_{\bsDelta})\le \begin{cases} 2\sqrt{d}\, \rho(\Lambda) & \text{if $h(\Lambda)\ge 1/2$,}\\ (1+\sqrt{d})\rho(\Lambda) & \text{otherwise.} \end{cases} \]
\end{remark}

\section{Construction of space-filling lattice designs}\label{sec:construction}

Based on the theoretical results in the previous section, our goal is to find a generating vector $\bsz$ that makes $\lambda_1(\Lambda)$ (or $\lambda_1(\Lambda^\perp)$) as large as possible.

\subsection{Explicit rank-1 lattice construction}\label{subsec:explicit}

For $d=2$, the Fibonacci lattice with $N=F_m$ and $\bsz=(1,F_{m-1})$ is a known explicit construction that achieves quasi-uniformity \cite{DGLPS25}. For $d\ge 3$, however, no explicit construction of good generating vectors is known so far. This problem is partly addressed in this subsection.

Let $\bsalpha=(\alpha_1,\ldots,\alpha_d)\in \RR^d$ be a fixed vector. Consider an infinite sequence of points of the form $\Scal=(\{n\bsalpha\})_{n\ge 1}$, called the \emph{Kronecker sequence}. In \cite{DGLPS25}, it is proven that the Kronecker sequence is quasi-uniform if $1,\alpha_1,\dots,\alpha_d\in \RR$ are linearly independent elements of an algebraic extension of the rational field $\QQ$ of degree $d+1$. Thus, for instance, if 
\begin{align}\label{eq:Kronecker_vector}
    \bsalpha = \left( p^{1/(d+1)}, p^{2/(d+1)},\dots, p^{d/(d+1)}\right), 
\end{align}
for a prime number $p$, the corresponding Kronecker sequence is quasi-uniform. The proof relies on approximating the initial segment of the Kronecker sequence by a rank-1 lattice point set. A similar proof technique was originally employed by Larcher \cite{L88}, who proved that such Kronecker sequences achieve an isotropic discrepancy of $O(N^{-1/d})$. It turns out that this approximation provides an explicit construction of good generating vectors, both in terms of quasi-uniformity and isotropic discrepancy.

Our construction method proceeds as follows:

\begin{algorithm}[ht]\caption{Explicit construction of rank-1 lattices via algebraic numbers}\label{alg:explicit}\begin{algorithmic}[1]
    \REQUIRE Dimension $d \ge 1$, algebraic vector $\bsalpha \in \RR^d$ such that $\{1, \alpha_1, \dots, \alpha_d\}$ is linearly independent over $\QQ$ and spans an algebraic extension of degree $d+1$.
    \STATE Set $N_1 = 1$.
    \FOR{$i = 2, 3, \dots$}
    \STATE \label{line3} Find the smallest integer $N > N_{i-1}$ such that $\langle N \bsalpha \rangle < \langle N_{i-1} \bsalpha \rangle$, where $\langle \bsx \rangle = \max_{1 \le j \le d} \min_{k_j \in \ZZ} |x_j - k_j|$ denotes the $\ell_\infty$-distance to the nearest integer vector.
    \STATE Set $N_i = N$.
    \STATE For each $j=1,\dots,d$, let $z_j$ be the nearest integer to $N_i\alpha_j$.
    \STATE \textbf{Output} the pair $(N_i, \bsz)$ with $\bsz = (z_1, \dots, z_d)$.
\ENDFOR\end{algorithmic}\end{algorithm}

The major advantage of this algorithm is that the generating vector $\bsz$ is determined explicitly by rounding each component of $N_i\bsalpha$ to the nearest integer. The disadvantage, however, is that the sequence of point numbers $N_i$ cannot be chosen arbitrarily, as it is strictly determined by the Diophantine properties of the chosen vector $\bsalpha$. This is analogous to the two-dimensional Fibonacci lattice, where the number of points is restricted to Fibonacci numbers.

The minimality of $N_i$ in Step~3 automatically ensures that $\gcd(N_i, z_1, \dots, z_d) = 1$; if there were a common divisor $g > 1$, then $N' = N_i/g$ and $z'_j=z_j/g$ ($j=1,\ldots,d$) would be integers, yielding a smaller error 
\[ \langle N' \bsalpha \rangle \le \max_{1\le j\le d}|N'\alpha_j-z'_j| = \frac{1}{g}\max_{1\le j\le d}|N\alpha_j-z_j| =\frac{1}{g} \langle N_i \bsalpha \rangle < \langle N_i \bsalpha \rangle, \] 
contradicting the fact that $N_i$ is the smallest integer achieving an error less than $\langle N_{i-1} \bsalpha \rangle$. Furthermore, it was proven in \cite{DGLPS25} that $N_{i+1}\le 2c^{-d}N_i$ for any $i\in \NN$, where $c$ is the constant appearing in \eqref{eq:badly_approximable} below. Thus, in the light of \cref{def:quasi-uniform_point_set}, this bounded growth of the sequence $(N_i)$ implies that to show \cref{alg:explicit} generates a quasi-uniform family of point sets, it suffices to prove that the mesh ratio of each rank-1 lattice point set $P_{N_i,\bsz}$ is bounded by a constant independent of $i$.

\begin{remark}
    The assumption on $\bsalpha$ ensures that $\bsalpha$ is badly approximable in the sense of simultaneous Diophantine approximation. That is, there exists a constant $0<c\le 1/2$, depending only on $\bsalpha$ and $d$, such that
    \begin{align}\label{eq:badly_approximable}
        \langle k \bsalpha \rangle \ge \frac{c}{k^{1/d}}\quad \text{for all $k\in \NN$.}
    \end{align}
    This property is a direct consequence of Schmidt's subspace theorem (cf.~\cite[Chapter~V, Theorem~1B]{S80}), and is fundamental in \cite{DGLPS25} to proving that the corresponding Kronecker sequence is quasi-uniform.
\end{remark}

Now let us consider any pair $(N_i,\bsz)$ obtained by \cref{alg:explicit}, and the corresponding generating matrix of the form \eqref{eq:generating_matrix} and its lattice $\Lambda$.
Following the results in \cite[Lemma~3.12]{DGLPS25}, the algebraic properties of $\bsalpha$ ensure that there exists a constant $c_1>0$, depending only on $\bsalpha$ and $d$, such that
\begin{align}\label{eq:lower_lambda_1}
    \lambda_1(\Lambda)\ge \frac{c_1}{N_i^{1/d}}.
\end{align}
Applying this lower bound to \cref{lem:rho_inf_bound}, we obtain the following result:

\begin{theorem}\label{thm:main1}
    The sequence of rank-1 lattice point sets $\{P_{N_i, \bsz}\}_{i >1}$ generated by \cref{alg:explicit} is quasi-uniform.
\end{theorem}

\begin{proof}
    Following \cref{lem:rho_inf_bound} together with the lower bound on $\lambda_1(\Lambda)$, the lattice $\Lambda$ corresponding to the rank-1 lattice point set $P_{N_i,\bsz}$ satisfies
    \begin{align*}
        \rho(\Lambda) \le \sqrt{d}\, C_d \frac{|\det(T)|}{(\lambda_1(\Lambda))^d}\le \sqrt{d}\, C_d \frac{1/N_i}{(c_1/N_i^{1/d})^d} = \frac{\sqrt{d}\, C_d}{c_1^d},
    \end{align*}
    where the upper bound is a constant independent of $N_i$. This uniform bound on $\rho(\Lambda)$, combined with the bounded growth of the sequence $(N_i)$ discussed above and \cref{lem:mesh_ratio_comparison}, ensures that the sequence of point sets $\{P_{N_i, \bsz}\}$ is quasi-uniform.
\end{proof}

As announced, we further show that the same sequence of rank-1 lattice point sets achieves an isotropic discrepancy of $O(N^{-1/d})$. First, let us recall the definition: for a point set $P\subset [0,1)^d$ with $N$ points, the isotropic discrepancy $J(P)$ is defined as
\[ J(P) = \sup_{C\in \Ccal}\left| \frac{|P\cap C|}{N}-\vol(C)\right|, \]
where $\Ccal$ denotes the family of all convex subsets of $[0,1)^d$. It is worth noting that the optimal order of the isotropic discrepancy is known to be $O(N^{-2/(d+1)})$ up to logarithmic factors; see \cite{S75} for the lower bound and \cite{B88,S77} for the upper bound. However, this rate is not achievable by lattice point sets. In fact, within the class of integration lattices, the rate of $O(N^{-1/d})$ is known to be the best possible; see \cite[Theorems~5.24 and 5.26]{DKP22}.

\begin{theorem}
    The rank-1 lattice point set $P_{N_i, \bsz}$ generated by \cref{alg:explicit} satisfies
    \[ J(P_{N_i,\bsz})\le \frac{C}{N_i^{1/d}},\]
    for some constant $C$ independent of $i$.
\end{theorem}

\begin{proof}
    Consider the lattice $\Lambda$ corresponding to the rank-1 lattice point set $P_{N_i,\bsz}$. By applying the lower bound for $\lambda_1(\Lambda)$ shown in \eqref{eq:lower_lambda_1} to the transference theorem \eqref{eq:transference}, we have
    \begin{align*}
        \lambda_d(\Lambda^{\perp})\le \frac{d}{\lambda_1(\Lambda)}\le \frac{d}{c_1}N_i^{1/d}.
    \end{align*}
    Using Minkowski's second theorem, and noting that $\lambda_1(\Lambda^{\perp})\le \cdots \le \lambda_d(\Lambda^{\perp})$, we obtain
    \begin{align*}
        \lambda_1(\Lambda^{\perp}) & \ge \frac{2^d\Gamma(1+d/2)}{d!\, \pi^{d/2}}\left| \det(T^{-1})\right| \prod_{j=2}^{d}\frac{1}{\lambda_j(\Lambda^{\perp})} \\
        & \ge \frac{2^d\Gamma(1+d/2)}{d!\, \pi^{d/2}}\frac{\left| \det(T^{-1})\right| }{(\lambda_d(\Lambda^{\perp}))^{d-1}}  \ge \frac{2^d\Gamma(1+d/2)}{d!\,\pi^{d/2}}\left( \frac{c_1}{d}\right)^{d-1} N_i^{1/d},
    \end{align*}
    where we used $\det(T^{-1}) = N$. 
    Finally, we employ the bound relating isotropic discrepancy to the dual lattice given in \cite[Theorem~5.22]{DKP22}. This leads to
    \begin{align*}
        J(P_{N_i,\bsz}) \le \frac{d\, 2^{2(d+1)}}{\lambda_1(\Lambda^{\perp})}\le \frac{d^d\, 2^{d+2}\, d!\,\pi^{d/2}}{c_1^{d-1}\Gamma(1+d/2)}\cdot \frac{1}{N_i^{1/d}},
    \end{align*}
    where the leading factor is a constant independent of $N_i$. This completes the proof.
\end{proof}

\subsection{Algorithmic Korobov lattice construction}\label{subsec:korobov}

The explicit construction given in the previous section has a practical limitation, as the sequence of point numbers $(N_i)_{i\in \NN}$ is strictly determined by the Diophantine properties of the algebraic vector $\bsalpha$. In many applications, however, it is desirable to construct a rank-1 lattice point set with a prescribed number of points to meet a specific computational budget.

To this end, here we consider the class of Korobov lattice point sets, which are a special case of rank-1 lattice point sets where the generating vector takes the form $\bsz=(1,a,\ldots,a^{d-1}) \pmod N$ for some parameter $a\in \{1,\ldots,N-1\}$. For a fixed $N$, a ``good''  parameter $a$ can be found by a brute-force search over the candidate set $\{1,\ldots,N-1\}$. To evaluate the quality of each $a$, one needs to compute the length of the shortest non-zero vector of the associated lattice $\Lambda$ or its dual lattice $\Lambda^{\perp}$ (cf. the bound on the mesh ratio shown in \cref{lem:rho_inf_bound}).

While the shortest vector problem is known to be NP-hard in general, the LLL lattice basis reduction algorithm \cite{LLL82,NS09,NV10} provides a high-quality approximation (and often the exact shortest vector) in polynomial time. As a related study, we refer to \cite{HE06}, which attempted to use the LLL algorithm to optimize Korobov lattice point sets in the context of high-dimensional integration.

The search procedure is summarized in \cref{alg:lll_search}. Regarding its computational complexity, the search involves $N-1$ iterations. In each iteration, the LLL algorithm is applied to $d$-dimensional lattices with matrix entries bounded by $N$. The original exact arithmetic formulation requires $O(d^6 (\log N)^3)$ arithmetic operations, whereas it can be improved to $O(d^4\log N)$ when using floating-point arithmetic \cite{NS09}. Consequently, the total cost of the search is almost linear in $N$ and polynomial in $d$, making the construction feasible for moderately high dimensions $d$ and large $N$.

\begin{algorithm}[ht]
\caption{LLL-based search for Korobov parameters}\label{alg:lll_search}
\begin{algorithmic}[1]
\REQUIRE Number of points $N$, dimension $d$.
\STATE Initialize $score_{\max} \gets 0$ and $a^* \gets 1$.
\FOR{$a = 1$ \TO $N-1$}
\STATE Construct the generating matrices $T$ and $(T^{-1})^{\top}$ for the lattice $\Lambda$ and its dual lattice $\Lambda^\perp$ as in \eqref{eq:generating_matrix} and \eqref{eq:dual_matrix}, respectively.
\STATE Apply the LLL algorithm to $T$ and $(T^{-1})^{\top}$ to obtain reduced bases.
\STATE Let $\bsv_1$ and $\bsv_1^\perp$ be the shortest vectors in the reduced bases.
\STATE $current\_score \gets \|\bsv_1\|_2 \cdot \|\bsv_1^\perp\|_2$.
\IF{$current\_score > score_{\max}$}
\STATE $score_{\max} \gets current\_score$.
\STATE $a^* \gets a$.
\ENDIF
\ENDFOR
\RETURN $a^*$.
\end{algorithmic}
\end{algorithm}

The theoretical analysis below is based on the bound 
\begin{align}\label{eq:bound_mesh_ratio_for_theory}
     \rho(\Lambda) \le d\, C_d \frac{N}{(\lambda_1(\Lambda^{\perp}))^d}
\end{align}
which stems from \cref{lem:rho_inf_bound} and the fact that $\det(T^{-1})=N$. 
This bound implies that, if there exists a parameter $a$ (within the candidate set) such that $\lambda_1(\Lambda^{\perp})=\Omega(N^{1/d})$, then $\rho(\Lambda)$ is bounded by a constant independent of $N$.
In the numerical implementation, as described in \cref{alg:lll_search}, however, we rely on the product $\lambda_1(\Lambda)\lambda_1(\Lambda^{\perp})$ as a robust quality measure. This product also serves as an upper bound on the mesh ratio:
\[ \rho(\Lambda) \le \sqrt{d} \frac{\lambda_d(\Lambda)}{\lambda_1(\Lambda)} \le \frac{d\sqrt{d}}{\lambda_1(\Lambda)\lambda_1(\Lambda^{\perp})}, \]
where we have applied the transference theorem \eqref{eq:transference}. By maximizing this product, we intend to achieve a favorable balance between the separation property (governed by $\lambda_1(\Lambda)$) and the covering property (governed by $\lambda_1(\Lambda^{\perp})$).

As a counterpart of \cref{thm:main1}, we now provide a theoretical guarantee that \cref{alg:lll_search} also generates a quasi-uniform family of point sets.

\begin{theorem}\label{thm:main2}
    Let $(N_i)_{i\in \NN}$ be a sequence of point numbers such that $N_i$ denotes the $i$-th smallest prime number. For each $N_i$, let the parameter $a_i\in \{1,\ldots,N_i-1\}$ for the Korobov lattice point set be the one found by \cref{alg:lll_search} for the dimension $d$. Then, the resulting sequence of Korobov lattice point sets is quasi-uniform.
\end{theorem}

\begin{proof}
According to Bertrand's postulate, $N_{i+1}<2N_i$ holds for all $i$, ensuring the growth condition required for the family of point sets to be quasi-uniform (see \cref{def:quasi-uniform_point_set}). Thus, it suffices to prove that the mesh ratio $\rho(\Lambda)$ remains bounded by a constant for any prime number $N$.

Based on the bound on the mesh ratio \eqref{eq:bound_mesh_ratio_for_theory}, our goal is to show that for each fixed prime $N$, there exists a parameter $a\in \{1,2,\ldots,N-1\}$ such that $\lambda_1 (\Lambda^{\perp})=\Omega(N^{1/d})$. We follow a counting argument similar to the one used in \cite{DGLPS25} for general rank-1 lattice point sets.

Using the relationship between the Euclidean norm and the $\ell_1$-norm, the length of the shortest non-zero vector in the dual lattice $\Lambda^{\perp}$ can be bounded as
\begin{align*}
    \lambda_1 (\Lambda^{\perp}) & = \min_{\bsx\in \Lambda^{\perp}\setminus \{\bszero\}}\|\bsx\|_2 \ge \frac{1}{\sqrt{d}} \min_{\bsx\in \Lambda^{\perp}\setminus \{\bszero\}}\|\bsx\|_1\\
    & = \frac{1}{\sqrt{d}}  \min_{\substack{\bsh\in \ZZ^d\setminus \{\bszero\}\\ \bsh\cdot (1,a,\ldots,a^{d-1})\equiv 0 \pmod N}}\|\bsh\|_1.
\end{align*}
Here, we used the fact that the dual lattice for the Korobov lattice $\Lambda$ can be explicitly written as
\[ \Lambda^{\perp} := \left\{ \bsh \in \ZZ^d \mid \bsh \cdot (1,a,\ldots,a^{d-1}) \equiv 0 \pmod{N} \right\}. \]

For any $\ell\in \NN_0$, the number of integer vectors with $\ell_1$-norm equal to $\ell$ satisfies
\[ \left| \left\{ \bsh\in \ZZ^d \mid \|\bsh\|_1=\ell \right\}\right| \le 2^d\binom{\ell+d-1}{d-1}, \]
from which we can bound the number of vectors within an $\ell_1$-ball of radius $\kappa$ as 
\[ \left| \left\{ \bsh\in \ZZ^d \mid \|\bsh\|_1\le \kappa \right\}\right| \le \sum_{\ell=0}^{\kappa}2^d\binom{\ell+d-1}{d-1}= 2^d\binom{\kappa+d}{d}.\]
For each $\bsh\in \ZZ^d\setminus \{\bszero\}$ such that $\|\bsh\|_1\le \kappa<N$, the congruence 
\[ h_1+h_2 a+h_3a^2+\cdots+h_d a^{d-1}\equiv 0 \pmod N \]
is a non-zero polynomial equation in $a$ over the field $\ZZ_N$ of degree at most $d-1$. Since $N$ is prime, there are at most $d-1$ choices of $a \in \{1, \dots, N-1\}$ that satisfy this condition. Consequently, the number of ``bad'' parameters $a$ for which $\min \|\bsh\|_1 \le \kappa$ is bounded by
\[ \left| \left\{ a\in \{1,\ldots,N-1\}\mid \min_{\substack{\bsh\in \ZZ^d\setminus \{\bszero\}\\ \bsh\cdot (1,a,\ldots,a^{d-1})\equiv 0 \pmod N}}\|\bsh\|_1 \le \kappa \right\}\right|\le (d-1)2^d\binom{\kappa+d}{d}. \]

Since there are $N-1$ possible values for $a$, at least one ``good'' parameter $a$ (i.e., one yielding $\min \|\bsh\|_1>\kappa$) exists if 
\begin{align}\label{eq:condition}
    (d-1)2^d\binom{\kappa+d}{d}\le \frac{2^d(\kappa+d)^d}{(d-1)!}< N-1.
\end{align} 
We choose $\kappa$ as follows:
\[ \kappa=\left\lceil 2^{-1}\left( (d-1)!(N-1)\right)^{1/d}\right\rceil -d-1,\]
This choice satisfies the condition \eqref{eq:condition} whenever $\kappa > 0$. Let $N_0$ be the threshold (dependent only on $d$) such that $\kappa > 0$ for all $N \ge N_0$. 
For $N \ge N_0$, we have $\kappa = \Omega(N^{1/d})$, which implies the existence of a parameter $a$ such that $\lambda_1(\Lambda^{\perp}) \ge \kappa/\sqrt{d} = \Omega(N^{1/d})$.
For the finite set of primes $N < N_0$, since $\Lambda^{\perp}$ is a full-rank lattice, $\lambda_1(\Lambda^{\perp})$ is strictly positive. Thus, the minimum of the ratio $\lambda_1(\Lambda^{\perp}) / N^{1/d}$ over this finite set of primes is also strictly positive.
Combining these two cases, there exists a constant $c(d) > 0$ independent of $N$ such that
\[ \lambda_1 (\Lambda^{\perp}) \ge c(d)N^{1/d} \]
holds for all primes $N$. Substituting this into \eqref{eq:bound_mesh_ratio_for_theory}, we conclude that $\rho(\Lambda)$ is uniformly bounded.

Finally, we justify the scoring criterion used in \cref{alg:lll_search}. Although the algorithm selects the parameter $a^*$ that maximizes the product $\lambda_1(\Lambda)\lambda_1(\Lambda^{\perp})$, this is essentially equivalent to maximizing $\lambda_1(\Lambda^{\perp})$ (up to constant factors). In fact, from the transference theorem \eqref{eq:transference} and Minkowski's second theorem, we have
\[
    \frac{1}{\lambda_1(\Lambda)\lambda_1(\Lambda^{\perp})} \le \frac{\lambda_d(\Lambda^{\perp})}{\lambda_1(\Lambda^{\perp})} \le C_d \frac{|\det(T^{-1})|}{(\lambda_1(\Lambda^{\perp}))^d} = C_d \frac{N}{(\lambda_1(\Lambda^{\perp}))^d},
\]
where $C_d = 2^d \Gamma(1+d/2) / \pi^{d/2}$. Since \cref{alg:lll_search} performs an exhaustive search, the selected $a^*$ satisfies:
\[ \frac{1}{\lambda_1(\Lambda_{a^*})\lambda_1(\Lambda^{\perp}_{a^*})}=\min_{a\in \{1,\ldots,N-1\}}\frac{1}{\lambda_1(\Lambda_{a})\lambda_1(\Lambda^{\perp}_{a})}\le \min_{a\in \{1,\ldots,N-1\}}C_d \frac{N}{(\lambda_1(\Lambda^{\perp}_a))^d}. \]
Thus, the above existence result implies that the reciprocal of the product for $a^*$ remains bounded by a constant, ensuring that the mesh ratio $\rho(\Lambda_{a^*})$ satisfies the quasi-uniformity condition.
\end{proof}

\begin{remark}
    Since the moduli $N_i$ are assumed to be prime in \cref{thm:main2}, for any choice of the parameter $a \in \{1, \dots, N_i-1\}$, the corresponding generating vector $\bsz$ automatically satisfies $\gcd(N_i, z_j) = 1$ for all $j = 1, \dots, d$. This condition ensures that any one-dimensional projection of $P_{N_i, \bsz}$ coincides with the set of equidistant points $\{k/N_i \mid k=0, \dots, N_i-1\}$. This property is characteristic of Latin hypercube designs \cite{HRDH11,J16,KIWINTSG25,MBC79}, ensuring that the points are fully stratified across all dimensions. This property does not necessarily hold for the explicit rank-1 lattice construction, since we only have $\gcd(N, z_1, \dots, z_d) = 1$, so that the one-dimensional projections may overlap.
\end{remark}

\section{Numerical experiments}\label{sec:experiments}

The theoretical analysis in the previous section established that our two construction methods provide sequences of quasi-uniform designs. In this section, we numerically examine their space-filling properties and the performance of Gaussian process regression using these sequences as sampling nodes. We compare our proposed designs against standard low-discrepancy sequences and random point sets, where we employ the \texttt{QMCPy} framework \cite{SHCRRJ26} to generate the Halton and Sobol' sequences. Regarding the explicit rank-1 lattice construction, we employ the generating vector $\bsalpha$ of the form \eqref{eq:Kronecker_vector} with parameter $p$. Unless otherwise stated, we set $p=2$. To ensure reproducibility, the code used in these experiments is available at \url{https://github.com/Naokikiki/Quasi-Uniform-Points}.

\subsection{Empirical properties of two constructions}\label{subsec:empirical}

\subsubsection{Construction parameters and the number of points}
We first examine how the choice of parameters and the number of points are related in each construction. For the explicit rank-1 lattice, the sequence of the number of points $N$ is determined by the recursive algorithm (\cref{line3} in \cref{alg:explicit}) and depends crucially on the Diophantine property of the vector $\bsalpha$ (in our case, the choice of the prime parameter $p$). \Cref{tab:explicit_n_2d} compares the number of points (listing the first 18 terms of the infinite sequence) obtained for $d=2$ using different primes $p \in \{2,3,5,7\}$, alongside the Fibonacci sequence, which corresponds to the classical Fibonacci lattice. The Fibonacci lattice exhibits a highly regular growth pattern, with the ratio $N_i/N_{i-1}$ converging to the golden ratio $\varphi \approx 1.618$. The proposed explicit construction shows more variable growth rates, with ratios ranging from approximately $1.1$ to over $10$ depending on $p$. We observe that this variability becomes more pronounced for larger values of $p$, as indicated by the increasing maximum ratios shown in \cref{tab:explicit_n_2d}. However, it is noteworthy that the median remains around $2$ across different values of $p$, indicating that the number of points typically doubles at each step.

\Cref{tab:explicit_n_d} illustrates the impact of dimensionality $d\in \{2,3,5,7\}$ on the sequence of the number of points for our base choice $p=2$. As the dimension $d$ increases, the algebraic relationships between the components of $\bsalpha$ become more complicated, leading to larger fluctuations in the consecutive ratios $N_i/N_{i-1}$. Nevertheless, the median ratio remains around $2$ across different values of $d$, indicating that, in the typical case, the number of points roughly doubles at each step, consistent with the behavior observed for varying $p$.

In contrast, the Korobov lattice construction offers the flexibility to choose any prime number as the number of points $N$. For this approach, our focus shifts to finding the optimal generating vector for a fixed $N$. \Cref{tab:korobov_optimal_a} shows the optimal parameter $a^*$ selected for each pair $(N,d)$ by minimizing the upper bound of the mesh ratio $d\sqrt{d}/(\lambda_1(\Lambda)\lambda_1(\Lambda^{\perp}))$ over the candidate set $a \in \{1,\ldots,N-1\}$. In this table, we select prime values for $N$ that are close to powers of 2. These optimal parameters are used throughout the subsequent experiments.

\begin{table}
    \footnotesize
    \centering
    \begin{tabular}{c|ccccc}
       Index $i$  & Fibonacci & $p=2$ & $p=3$ & $p=5$ & $p=7$ \\ \hline
       $1$  & $1$ & $1$ & $1$  & $1$  & $1$ \\
       $2$  & $2$ & $3$ & $2$ & $3$ & $2$ \\
       $3$  & $3$ & $7$ & $11$ & $11$ & $3$ \\
       $4$  & $5$ & $12$ & $25$ & $14$ & $9$ \\
       $5$  & $8$ & $46$ & $113$ & $131$ & $12$ \\
       $6$  & $13$ & $177$ & $312$ & $224$ & $35$ \\
       $7$  & $21$ & $681$ & $1411$ & $500$ & $126$ \\
       $8$  & $34$ & $858$ & $2485$ & $1224$ & $138$ \\
       $9$  & $55$ & $2620$ & $3896$ & $1724$ & $1447$ \\
       $10$  & $89$ & $5921$ & $9515$ & $13161$ & $1585$ \\
       $11$  & $144$ & $10080$ & $21515$ & $14385$ & $16619$ \\
       $12$  & $233$ & $22780$ & $48649$ & $16109$ & $18204$ \\
       $13$  & $377$ & $38781$ & $70164$ & $45379$ & $190872$ \\
       $14$  & $610$ & $149203$ & $268656$ & $212010$ & $209076$ \\
       $15$  & $987$ & $574032$ & $607476$ & $574542$ & $1792249$ \\
       $16$  & $1597$ & $723235$ & $3354685$ & $1406473$ & $2192197$ \\
       $17$  & $2584$ & $2208486$ & $7585502$ & $1981015$ & $2401273$ \\
       $18$  & $4181$ & $2782518$ & $41889671$ & $5580513$ & $4593470$ \\ \hline \hline
       $\min N_i/N_{i-1}$ & $1.5000$ & $1.2599$ & $1.4422$ & $1.0930$ & $1.0952$ \\
       $\mathrm{median}\, N_i/N_{i-1}$ & $1.6180$ & $2.3333$ & $2.2727$ & $2.4480$ & $1.9129$ \\
       $\max N_i/N_{i-1}$ & $2$ & $3.8478$ & $5.5223$ & $9.3571$ & $10.4855$ \\
    \end{tabular}
    \caption{Sequence of the number of points for $d=2$: Comparison between the explicit construction with different primes $p$ and the Fibonacci sequence.}
    \label{tab:explicit_n_2d}
\end{table}

\begin{table}
    \footnotesize
    \centering
    \begin{tabular}{c|ccccc}
       Index $i$  & $d=2$ & $d=3$ & $d=5$ & $d=7$ \\ \hline
       $1$  & $1$ & $1$ & $1$  & $1$ \\
       $2$  & $3$ & $2$ & $3$ & $2$ \\
       $3$  & $7$ & $4$ & $8$ & $4$ \\
       $4$  & $12$ & $7$ & $15$ & $26$ \\
       $5$  & $46$ & $9$ & $24$ & $31$ \\
       $6$  & $177$ & $10$ & $31$ & $343$ \\
       $7$  & $681$ & $22$ & $65$ & $1813$ \\
       $8$  & $858$ & $31$ & $138$ & $1977$ \\
       $9$  & $2620$ & $53$ & $531$ & $2156$ \\
       $10$  & $5921$ & $63$ & $596$ & $3790$ \\
       $11$  & $10080$ & $116$ & $669$ & $20031$ \\
       $12$  & $22780$ & $333$ & $6883$ & $23821$ \\
       $13$  & $38781$ & $1480$ & $8730$ & $74744$ \\
       $14$  & $149203$ & $1760$ & $9326$ & $221318$ \\
       $15$  & $574032$ & $3969$ & $9995$ & $241349$ \\
       $16$  & $723235$ & $5729$ & $13662$ & $267326$ \\
       $17$  & $2208486$ & $7822$ & $14862$ & $2929805$ \\
       $18$  & $2782518$ & $13155$ & $31544$ & $2953626$ \\ \hline \hline
       $\min N_i/N_{i-1}$ & $1.2599$ & $1.1111$ & $1.0683$ & $1.0081$ \\
       $\mathrm{median}\, N_i/N_{i-1}$ & $2.3333$ & $1.7097$ & $1.6000$ & $2.0000$ \\
       $\max N_i/N_{i-1}$ & $3.8478$ & $4.4444$ & $10.2885$ & $11.0645$ \\
    \end{tabular}
    \caption{Sequence of the number of points for $p=2$: Comparison across different dimensions $d$.}
    \label{tab:explicit_n_d}
\end{table}

\begin{table}
    \footnotesize
    \centering
    \begin{tabular}{c|ccccc}
       Number of points $N$  & $d=2$ & $d=3$ & $d=5$ & $d=7$ \\ \hline
       $3$ & $1$ & $2$ & $1$ & $1$ \\
        $7$ & $3$ & $2$ & $2$ & $2$ \\
        $13$ & $8$ & $2$ & $2$ & $2$ \\
        $31$ & $12$ & $17$ & $2$ & $24$ \\
        $61$ & $53$ & $15$ & $49$ & $16$ \\
        $127$ & $115$ & $102$ & $82$ & $11$ \\
        $251$ & $177$ & $106$ & $124$ & $192$ \\
        $509$ & $376$ & $225$ & $20$ & $354$ \\
        $1021$ & $798$ & $516$ & $916$ & $461$ \\
        $2039$ & $1062$ & $131$ & $1939$ & $1140$ \\
        $4093$ & $1127$ & $3506$ & $742$ & $2091$ \\
        $8191$ & $6725$ & $5605$ & $7349$ & $3457$ \\
       $16381$  & $6259$ & $11599$ & $932$ & $7000$ \\
       $32749$  & $29342$ & $26709$ & $2534$ & $24511$ \\
       $65521$  & $6385$ & $45077$ & $28529$ & $37400$ \\
       $131071$  & $80509$ & $107018$ & $36684$ & $23796$ \\
       $262139$  & $193187$ & $197360$ & $104738$ & $172874$ \\
       $524287$  & $105982$ & $341893$ & $307592$ & $58346$ \\
    \end{tabular}
    \caption{Optimal generating parameters $a^*$ for Korobov lattice point sets for selected primes $N$ and dimensions $d$.}
    \label{tab:korobov_optimal_a}
\end{table}

\subsubsection{Bound on mesh ratio}

We investigate the behavior of the mesh ratio as the number of points increases. Although the theoretical results in the previous section show that the mesh ratio of both constructions is uniformly bounded with respect to $N$, it is instructive to examine the constant factors involved. Since evaluating the mesh ratio exactly is computationally prohibitive, in \cref{fig:mesh_ratio_bounds}, we numerically 
show how the theoretical upper bound of the mesh ratio, given by $
d\sqrt{d}/(\lambda_1(\Lambda)\lambda_1(\Lambda^{\perp}))$, behaves as a function of the number of points $N$ (with the horizontal axis on a logarithmic scale) for dimensions $d\in \{2,3,5,7\}$.

For the numerical evaluation of the explicit lattice construction, we note a technical detail regarding the basis construction. Although \cref{alg:explicit} ensures that the components are collectively coprime to $N$, i.e., $\gcd(N,z_1,\ldots,z_d)=1$, as discussed in \cref{subsec:explicit}, there is no guarantee that there exists a single index $j$ such that $\gcd(N,z_j)=1$. In such cases, the generating matrix cannot be represented in the standard form \eqref{eq:generating_matrix}. To compute the shortest non-zero vector length $\lambda_1(\Lambda)$ via LLL reduction, we instead construct the basis as follows: we start with the $d \times (d+1)$ matrix $[\bsz^{\top} \mid N\mathbf{I}_d]$ and apply unimodular column operations to reduce it to a $d\times d$ matrix. The resulting matrix is then scaled by $1/N$ to obtain a basis for $\Lambda$ to which the LLL reduction is applied.

We observe that the Korobov construction generally exhibits more stable and lower bounds across different values of $N$ and $d$, benefiting from the optimization of the parameter $a$. In contrast, the explicit construction shows greater variability and tends to yield larger bound values. Note that a lower upper bound does not strictly imply a smaller actual mesh ratio, but it does provide a tighter theoretical guarantee. Furthermore, as expected, the magnitude of these bounds increases with the dimension $d$, primarily due to the growth of the factor $d\sqrt{d}$ in the numerator. 

\begin{figure}[t]
    \centering

    \begin{minipage}[b]{0.45\textwidth}
    \centering
    \includegraphics[width=\linewidth]{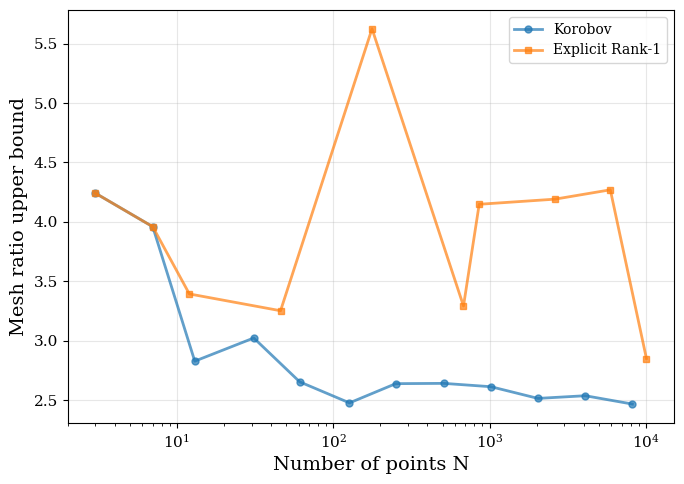}
    \subcaption{$d=2$}
    \label{fig:bound_2d}
    \end{minipage}
    \begin{minipage}[b]{0.45\textwidth}
    \centering
    \includegraphics[width=\linewidth]{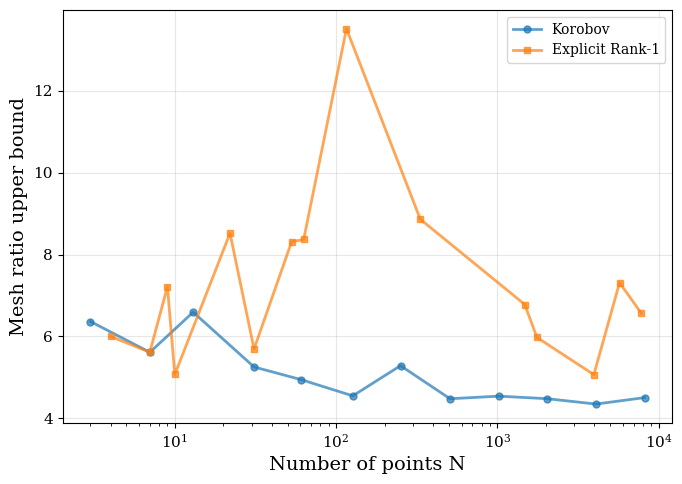}
    \subcaption{$d=3$}
    \label{fig:bound_3d}
    \end{minipage}

    \begin{minipage}[b]{0.45\textwidth}
    \centering
    \includegraphics[width=\linewidth]{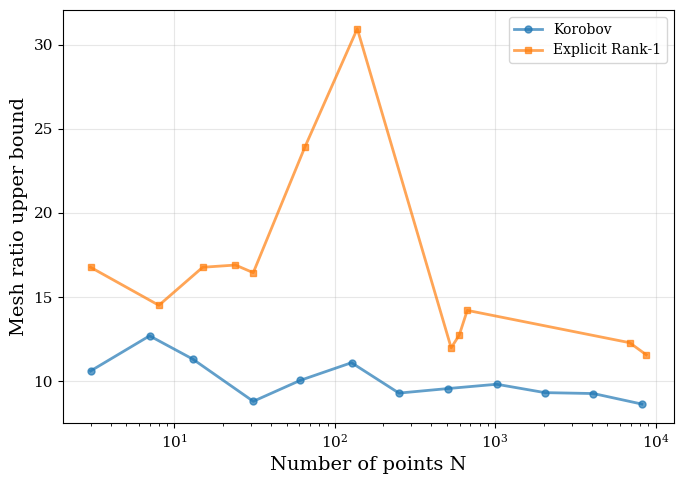}
    \subcaption{$d=5$}
    \label{fig:bound_5d}
    \end{minipage}
    \begin{minipage}[b]{0.45\textwidth}
    \centering
    \includegraphics[width=\linewidth]{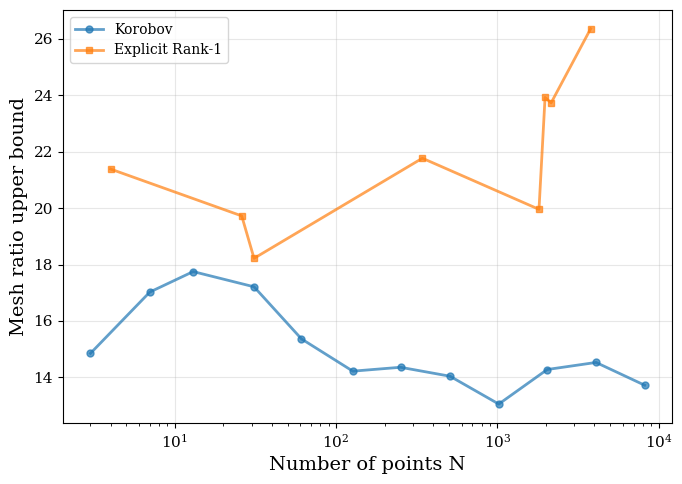}
    \subcaption{$d=7$}
    \label{fig:bound_7d}
    \end{minipage}
    \caption{The upper bound of the mesh ratio as a function of the number of points $N$ for various dimensions.}
    \label{fig:mesh_ratio_bounds}
\end{figure}

\subsubsection{Separation radius}

Quasi-uniform sequences are characterized by the optimal decay rate $\Theta(N^{-1/d})$ with respect to both the covering and separation radii. While standard low-discrepancy sequences, such as Halton and Sobol' sequences, are known to achieve the optimal decay rate for the covering radius (see, for instance, \cite[Chapter~6]{N92}), they typically fail to satisfy the optimality for the separation radius (see \cite{G24a} for the two-dimensional Sobol' sequence and \cite{GHS25} for the Halton sequence). Consequently, comparing the covering radius would primarily reveal differences in constant factors. In contrast, the separation radius provides a fundamental distinction between quasi-uniform point sets and those low-discrepancy sequences. Furthermore, while computing the exact covering radius of a given point set is computationally demanding, the separation radius can be exactly evaluated by computing the minimum of all pairwise distances. For these reasons, we numerically examine the behavior of the separation radius to verify the quasi-uniformity of our constructions, and compare them against Halton and Sobol' sequences as well as randomly sampled points.

The number of points used in the experiments depends on the type of point set. Because Sobol' sequences perform optimally when the number of points is a power of two, we likewise choose powers of two for the Halton sequence and for the randomly sampled points.
For the Korobov construction, we use the number of points and the corresponding parameters from \cref{tab:korobov_optimal_a}, which are prime numbers close to powers of two. The number of points for the explicit rank-1 lattice construction is determined by the construction algorithm, depending on the choice of $\bsalpha$, and is therefore not controllable. All experiments are conducted in dimensions $d = 2, 3, 5$, and $7$.

\Cref{fig:separation_radius} displays the decay of the separation radius $q(P_N)$ as a function of the number of points $N$ on a log-log scale. The results clearly show that both our explicit rank-1 lattice and Korobov lattice follow the theoretical decay rate $\Theta(N^{-1/d})$ (indicated by the slope of the reference lines) across all tested dimensions. In contrast, the Halton and Sobol' sequences, as well as random points, exhibit a significantly faster decay, indicating that the minimum distance between points shrinks much more rapidly than the optimal rate. This empirically confirms that, unlike the proposed lattice constructions, these standard sequences do not possess the quasi-uniform property.

\begin{figure}[t]
    \centering

    \begin{minipage}[b]{0.45\textwidth}
    \centering
    \includegraphics[width=\linewidth]{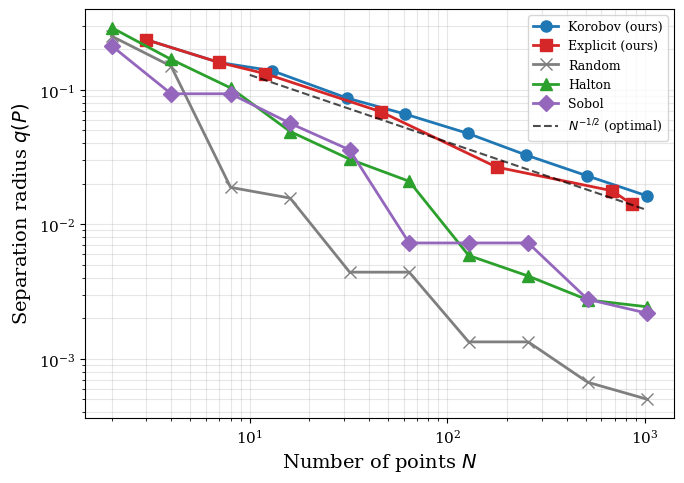}
    \subcaption{$d=2$}
    \label{fig:separation_2d}
    \end{minipage}
    \begin{minipage}[b]{0.45\textwidth}
    \centering
    \includegraphics[width=\linewidth]{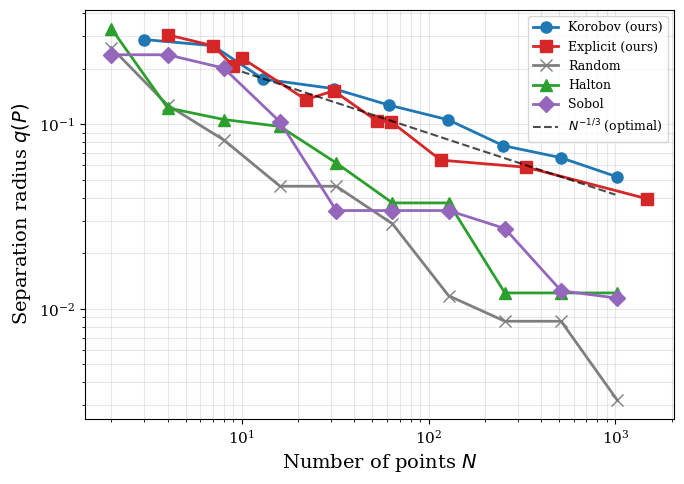}
    \subcaption{$d=3$}
    \label{fig:separation_3d}
    \end{minipage}

    \begin{minipage}[b]{0.45\textwidth}
    \centering
    \includegraphics[width=\linewidth]{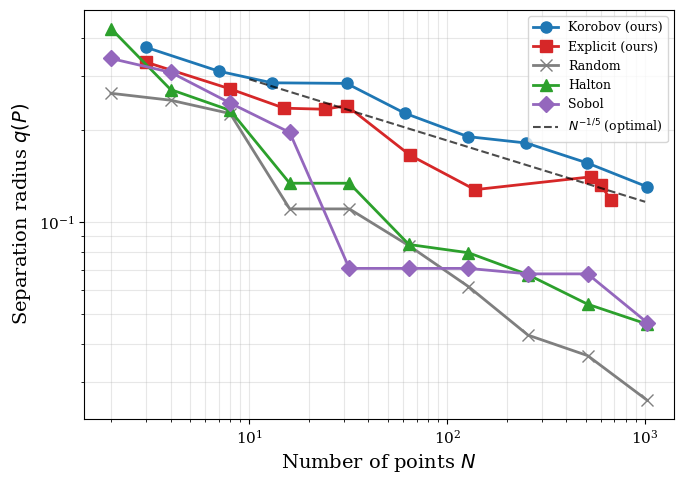}
    \subcaption{$d=5$}
    \label{fig:separation_5d}
    \end{minipage}
    \begin{minipage}[b]{0.45\textwidth}
    \centering
    \includegraphics[width=\linewidth]{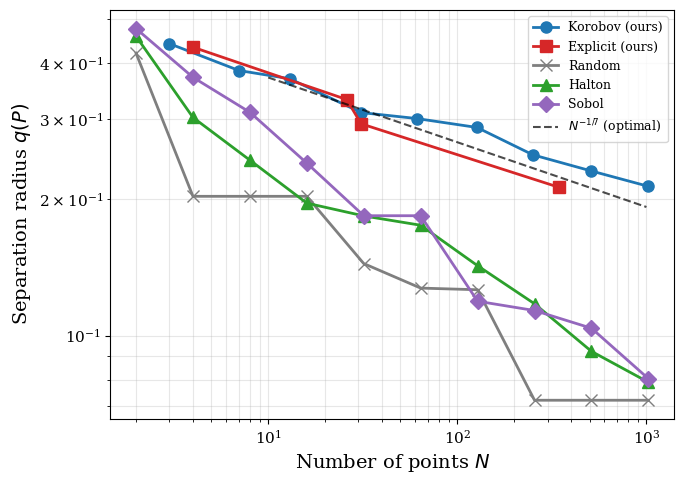}
    \subcaption{$d=7$}
    \label{fig:separation_7d}
    \end{minipage}

    \caption{Decay of the separation radius $q(P_N)$ as a function of the number of points $N$ for various dimensions. The proposed lattice sets maintain the optimal rate $\Theta(N^{-1/d})$, whereas other sequences show faster decay due to local clustering.}
    \label{fig:separation_radius}
\end{figure}

\subsection{Gaussian process regression}\label{subsec:GPR}

We evaluate the effectiveness of our point sets as training data for Gaussian process regression (GPR) \cite{RW06}. GPR is a widely used non-parametric method for function approximation and surrogate modeling, where the quality of predictions depends critically on the spatial distribution of training points. Since GPR prediction uncertainty grows in regions far from training points, minimizing the covering radius is essential for ensuring high approximation accuracy. Simultaneously, a large separation radius is required to avoid numerical instability (ill-conditioning) of the covariance matrix. Quasi-uniform point sets are particularly well-suited for this purpose \cite{T20,WBG21}.

\subsubsection{Experimental setup}

We compare the same five types of point sets as in the previous subsection: the Korobov lattice, the explicit rank-1 lattice, the Halton sequence, the Sobol' sequence, and randomly sampled points. For the GPR model, we use a Mat\'{e}rn kernel with smoothness parameter $\nu = 2.5$. We employ this fixed smoothness parameter across all test functions, regardless of their actual regularity, to assess the robustness of the designs under potential model misspecification. The kernel is defined as:
\begin{align*}
    k(\bsx, \bsy) = \sigma^2 \left(1 + \frac{\sqrt{5}\|\bsx - \bsy\|_2}{\ell} + \frac{5\|\bsx - \bsy\|_2^2}{3\ell^2}\right) \exp\left(-\frac{\sqrt{5}\|\bsx - \bsy\|_2}{\ell}\right),
\end{align*}
where the length scale $\ell$ and variance $\sigma^2$ are optimized by maximizing the marginal likelihood. To simulate realistic scenarios with measurement error, we add independent Gaussian noise with standard deviation $\sigma_{\mathrm{noise}} = 0.05$ to the training observations.

The approximation quality is measured by the $L^2$ error:
\begin{align*}
    \|f - \widehat{f}_{\mathrm{GPR}}\|_{L^2([0,1]^d)} = \left( \int_{[0,1]^d} |f(\bsx) - \widehat{f}_{\mathrm{GPR}}(\bsx)|^2 \, \mathrm{d}\bsx \right)^{1/2},
\end{align*}
where $f$ is the true function and $\widehat{f}_{\mathrm{GPR}}$ is the GPR posterior mean function. For $d = 2$, we compute this integral using tensor-product Gauss--Legendre quadrature with $50$ nodes per dimension. For $d \in \{3,5,7\}$, we employ Monte Carlo integration with $10{,}000$ random samples. All experiments are repeated over $100$ independent noise realizations, and we report the mean error along with the standard deviation.

\subsubsection{Test functions}

For the two-dimensional case, we use the Franke function~\cite{F07,F82}, a standard benchmark for scattered data interpolation:
\begin{align*}
    f(x, y) &= \frac{3}{4} \exp\left(-\frac{(9x-2)^2 + (9y-2)^2}{4}\right) + \frac{3}{4} \exp\left(-\frac{(9x+1)^2}{49} - \frac{9y+1}{10}\right) \\
    &\quad + \frac{1}{2} \exp\left(-\frac{(9x-7)^2 + (9y-3)^2}{4}\right) - \frac{1}{5} \exp\left(-(9x-4)^2 - (9y-7)^2\right).
\end{align*}
In addition, we evaluate three test functions adapted from benchmarks in \cite{LF03}. Since the original functions were defined on the unit disk centered at the origin, we map the domain to the unit square $[0,1]^2$ via the coordinate transformations $x \mapsto 2x-1$ and $y \mapsto 2y-1$. The resulting expressions are:
\begin{itemize}
    \item Arctan Function: $f(x, y) = \arctan\bigl(2(2x + 6y - 5)\bigr) / \arctan\bigl(2(\sqrt{10} + 1)\bigr)$.
    \item Gaussian Peak: $f(x, y) = \exp\left(-(2x-1.1)^2 - 0.5(2y-1)^2\right)$.
    \item Oscillatory Function: $f(x, y) = \sin\left(\pi\left((2x-1)^2 + (2y-1)^2\right)\right)$.
\end{itemize}

For the three-dimensional case, we use two functions from the Genz family~\cite{G84}, which are widely used benchmarks for numerical integration and surrogate modeling:
\begin{itemize}
    \item Genz Continuous: $f(\bsx) = \exp\left(-\sum_{i=1}^{d} c_i |x_i - 0.5|\right)$ with $\boldsymbol{c} = (1.5, 2.0, 2.5)^\top$.
    \item Genz Gaussian: $f(\bsx) = \exp\left(-\sum_{i=1}^{d} c_i^2 (x_i - 0.5)^2\right)$ with $\boldsymbol{c} = (3.0, 3.0, 3.0)^\top$.
\end{itemize}

For dimensions $d\in \{5,7\}$, we employ the pairwise trigonometric function:
\begin{align*}
    f(\bsx) = \frac{1}{d(d-1)/2} \sum_{1 \leq i < j \leq d} \sin(\pi(x_i + x_j)).
\end{align*}
This function consists of $\binom{d}{2}$ pairwise interaction terms ($10$ terms for $d=5$ and $21$ terms for $d=7$), normalized by the total number of pairs. This function emphasizes pairwise interactions and is sensitive to uniformity in two-dimensional projections.

\subsubsection{Results}

\paragraph{Two-dimensional results}
\Cref{fig:GPR_2d} presents the $L^2$ error as a function of the number of training points $N$ for the four two-dimensional test functions on a log-log scale. The results exhibit a consistent trend: the $L^2$ error decreases monotonically as $N$ increases. Notably, for sufficiently large sample sizes ($N \ge 100$), our two lattice constructions (Korobov and explicit rank-1) outperform the other point sets, demonstrating superior convergence behaviors.

\begin{figure}[t]
    \centering

    \begin{minipage}[b]{0.45\textwidth}
    \centering
    \includegraphics[width=\linewidth]{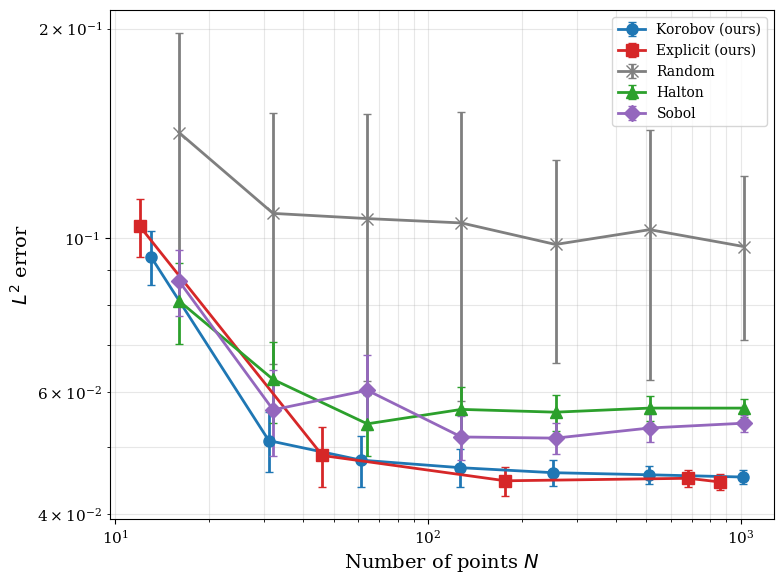}
    \subcaption{Franke}
    \label{fig:GPR_2d_Franke}
    \end{minipage}
    \begin{minipage}[b]{0.45\textwidth}
    \centering
    \includegraphics[width=\linewidth]{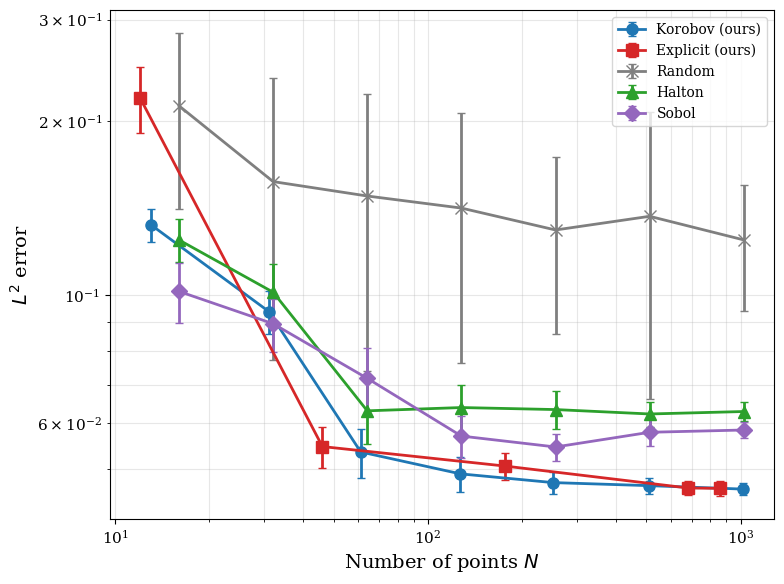}
    \subcaption{Arctan}
    \label{fig:GPR_2d_arctan}
    \end{minipage}

    \begin{minipage}[b]{0.45\textwidth}
    \centering
    \includegraphics[width=\linewidth]{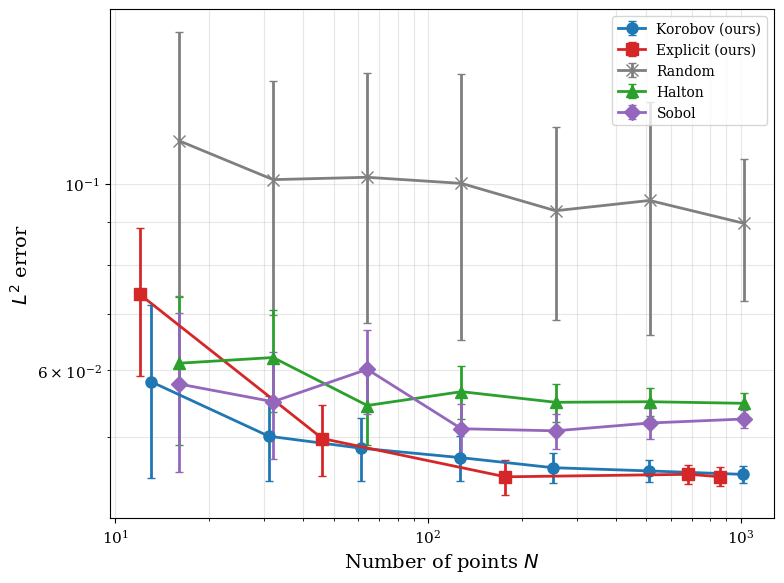}
    \subcaption{Gaussian Peak}
    \label{fig:GPR_2d_Gaussian}
    \end{minipage}
    \begin{minipage}[b]{0.45\textwidth}
    \centering
    \includegraphics[width=\linewidth]{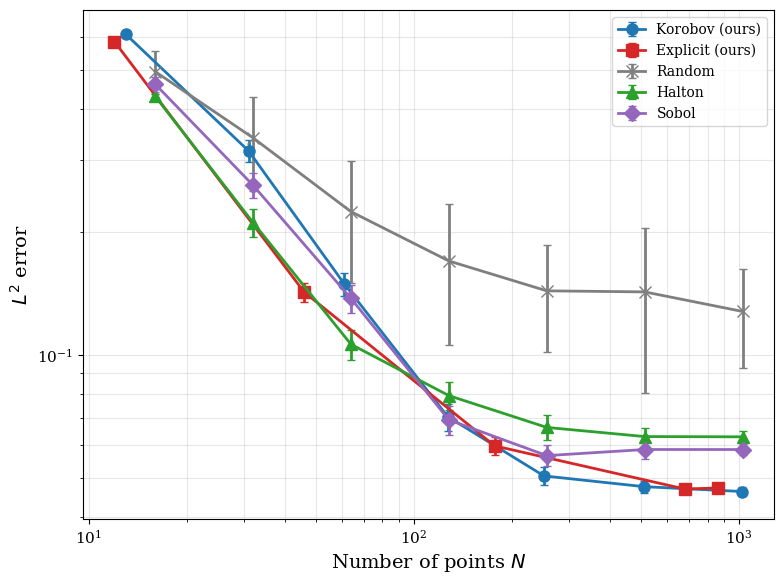}
    \subcaption{Oscillatory}
    \label{fig:GPR_2d_oscillatory}
    \end{minipage}

    \caption{GPR approximation error for the four test functions in $d=2$. The plots show the mean $L^2$ error over $100$ trials with observation noise $\sigma_{\mathrm{noise}} = 0.05$ on a log-log scale. Error bars indicate $\pm 1$ standard deviation.}
    \label{fig:GPR_2d}
\end{figure}

\paragraph{Three-dimensional results}
\Cref{fig:GPR_3d} displays the $L^2$ error trajectories for the two Genz test functions in three dimensions. Although the advantages of the lattice constructions are not as pronounced as in the two-dimensional case, they asymptotically yield smaller mean $L^2$ errors, effectively outperforming the other point sets.

\begin{figure}[t]
    \centering
    \begin{minipage}[b]{0.45\textwidth}
        \centering
        \includegraphics[width=\linewidth]{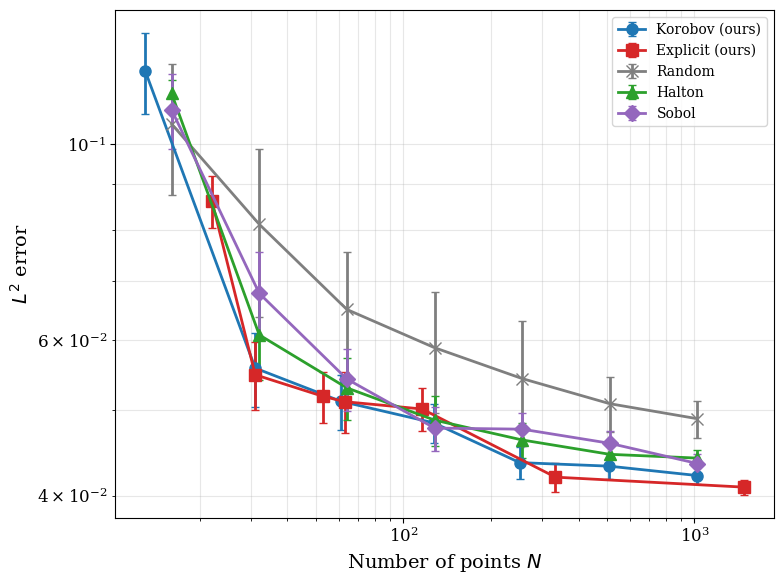}
        \subcaption{Genz Continuous}
        \label{fig:GPR_3d_continuous}
    \end{minipage}
    \begin{minipage}[b]{0.45\textwidth}
        \centering
        \includegraphics[width=\linewidth]{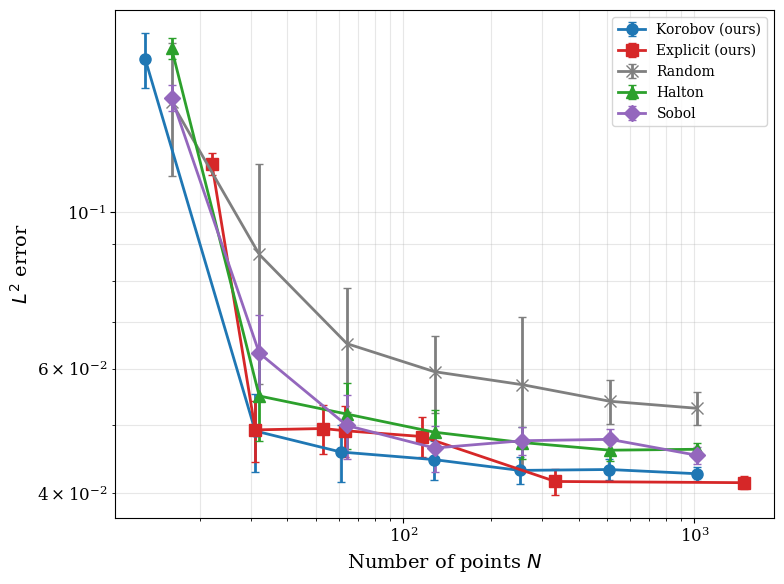}
        \subcaption{Genz Gaussian}
        \label{fig:GPR_3d_Gaussian}
    \end{minipage}

    \caption{GPR approximation error for the two test functions in $d=3$. The plots show the mean $L^2$ error over $100$ trials with observation noise $\sigma_{\mathrm{noise}} = 0.05$ on a log-log scale. Error bars indicate $\pm 1$ standard deviation.}
    \label{fig:GPR_3d}
\end{figure}

\paragraph{Higher-dimensional results}
\Cref{fig:GPR_high-dim} presents the results for the pairwise trigonometric function in dimensions $d = 5$ and $d = 7$. In these settings, our lattice constructions do not exhibit a distinct advantage over the other point sets within the considered range of sample sizes, although a marginal performance gain is observed for the lattice methods at the largest sample size in $d = 5$. 

This lack of significant differentiation is likely attributable to the slow decay of the covering radius in higher dimensions, which scales as $N^{-1/d}$. Due to this slow convergence rate, the sample sizes used in our experiments likely correspond to the pre-asymptotic regime, where the geometric advantages of quasi-uniform designs are not yet dominant. Unfortunately, extending the experimental range to the asymptotic regime is practically limited by the computational complexity of fitting GPR models, which scales cubically with the number of training points. Nevertheless, as demonstrated in \cref{subsec:empirical}, the separation radius of our lattice constructions achieves the optimal rate, whereas the other point sets do not. Thus, even in high-dimensional settings where approximation performance is comparable, the advantage of our lattice designs in terms of numerical stability should still hold.

Additionally, we observe a highly peculiar convergence behavior for the explicit rank-1 lattice in the $d=5$ case. This anomaly appears to be directly correlated with the geometric properties of the generated point sets. That is, in the region where the error is unusually large (around $N \approx 100$), the bound on the mesh ratio shown in \cref{fig:bound_5d} also exhibits large values. This suggests that the degradation in quasi-uniformity directly negatively impacts the approximation performance.

A closer inspection reveals structural issues in specific two-dimensional projections. For instance, when $N=138$, the resulting generating vector is $\bsz=(17,36,57,81,108)$. Considering the projection onto the second and fifth dimensions, we find that $\gcd(N, z_2, z_5) = \gcd(138, 36, 108) = 6$. Consequently, the lattice collapses onto only $138/6 = 23$ distinct points in this two-dimensional projection, causing significant overlaps. Even in cases where strict overlaps do not occur (i.e., the points remain distinct), we frequently observe that the points align on a small number of parallel lines in two-dimensional projections. This alignment leaves large regions of the domain unsampled, creating significant gaps. Since the pairwise trigonometric test function is composed of a sum of two-dimensional interaction terms, it is particularly sensitive to the quality of these projections. Whether due to overlaps or poor distribution on a few hyperplanes, the lattice fails to capture the variation of the function in specific subspaces, leading to the observed spikes in error.

This finding highlights a potential limitation of using fixed parameters; thus, optimizing the parameter $\bsalpha$ in the explicit rank-1 construction to avoid such degenerate projections is a promising direction for future research.

\begin{figure}[t]
    \centering
    \begin{minipage}[b]{0.45\textwidth}
        \centering
        \includegraphics[width=\linewidth]{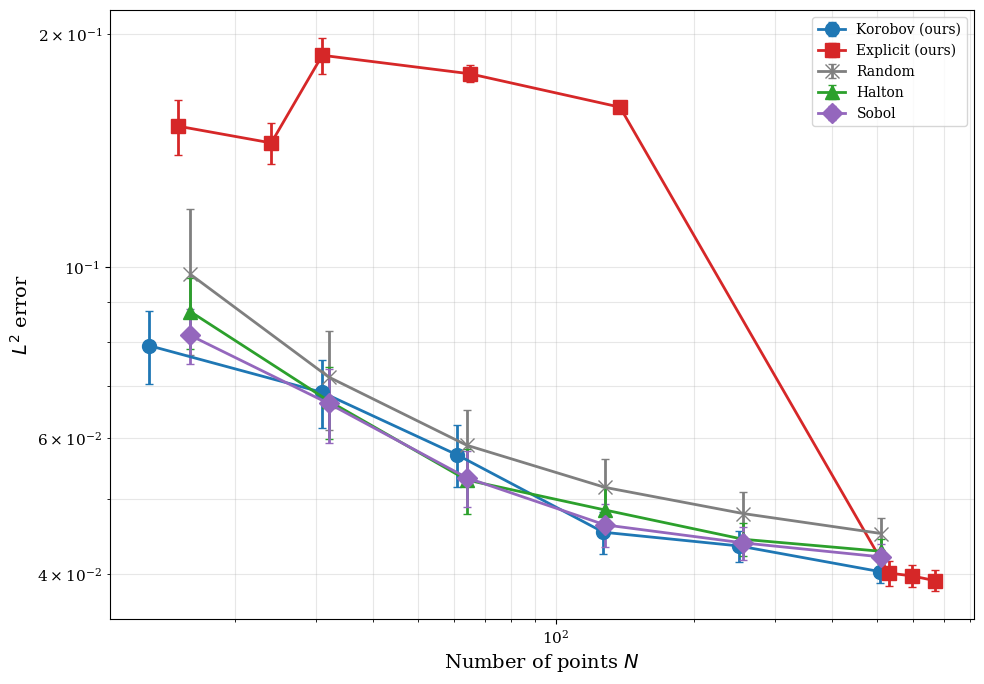}
        \subcaption{$d=5$}
        \label{fig:GPR_5d}
    \end{minipage}
    \begin{minipage}[b]{0.45\textwidth}
        \centering
        \includegraphics[width=\linewidth]{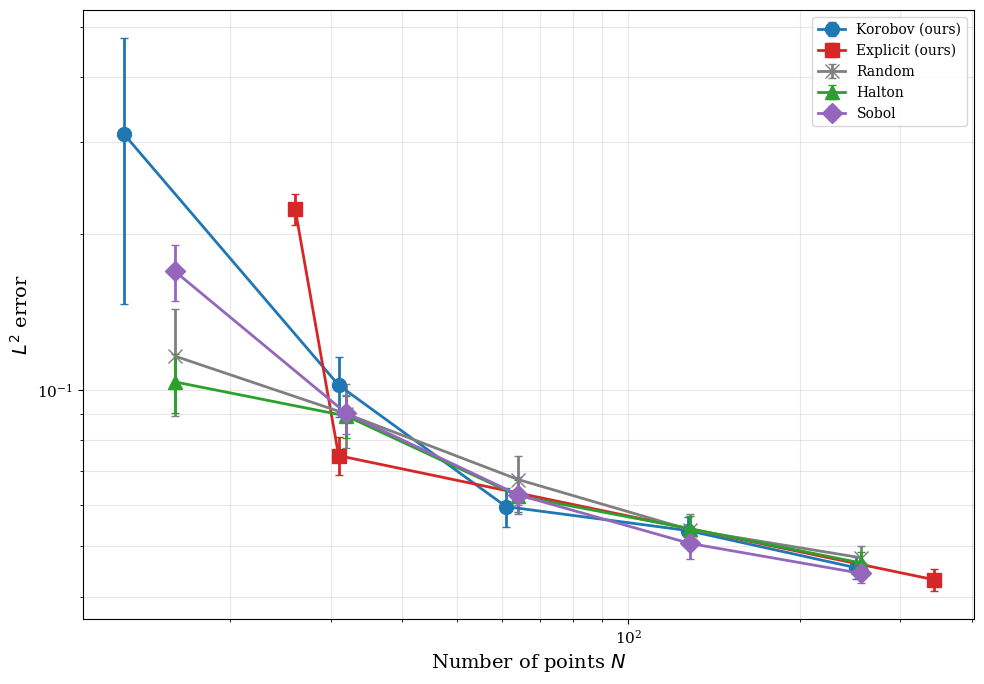}
        \subcaption{$d=7$}
        \label{fig:GPR_7d}
    \end{minipage}

    \caption{GPR approximation error for the pairwise trigonometric function in $d=5$ and $d=7$. The plots show the mean $L^2$ error over $100$ trials with observation noise $\sigma_{\mathrm{noise}} = 0.05$ on a log-log scale. Error bars indicate $\pm 1$ standard deviation.}
    \label{fig:GPR_high-dim}
\end{figure}

\section{Concluding remarks}\label{sec:conclusion}

In this paper, we have investigated the space-filling properties of rank-1 lattice point sets from the perspective of the geometry of numbers. While these point sets have often been analyzed in terms of their discrepancy properties, our present work highlights their utility as robust designs for computer experiments and scattered data approximation.

We proposed two construction algorithms: an explicit algebraic construction based on Diophantine approximation, and a greedy optimization of the Korobov lattice parameter using the LLL algorithm. Our theoretical analysis established that both constructions yield quasi-uniform families of point sets over the multi-dimensional unit cube.

Numerical experiments verified these theoretical findings. Empirically, the LLL-based construction achieves a smaller bound on the mesh ratio, demonstrating the benefit of optimizing the Korobov parameter, albeit with a computational cost that scales polynomially with the dimension $d$ and quasi-linearly with the number of points $N$. Crucially, our lattice designs maintain a separation radius of $\Theta(N^{-1/d})$. In contrast, random point sets and popular low-discrepancy sequences (such as Sobol' and Halton) exhibit a much faster decay in separation radius. This property makes rank-1 lattices particularly advantageous for kernel-based approximation methods, as it mitigates the severe ill-conditioning of kernel matrices possibly encountered with other point sets.

In terms of approximation performance in Gaussian process regression, our designs demonstrated a clear advantage in lower-dimensional settings ($d=2, 3$). In higher dimensions ($d=5, 7$), while the advantages in approximation error were less pronounced due to the limitations of sample size in the pre-asymptotic regime, the theoretical guarantee of numerical stability remains a key merit. Furthermore, our analysis of the explicit rank-1 construction in five dimensions uncovered a specific vulnerability: fixed algebraic parameters can lead to degenerate two-dimensional projections where points overlap or align on a few hyperplanes. This finding suggests an important direction for future research: optimizing the algebraic parameter $\bsalpha$ to maximize the quality of low-dimensional projections, which would enhance the robustness of the construction against structural anomalies.

\bibliographystyle{siamplain}
\bibliography{ref}

@article {B93,
    AUTHOR = {Banaszczyk, W.},
     TITLE = {New bounds in some transference theorems in the geometry of
              numbers},
   JOURNAL = {Math. Ann.},
  FJOURNAL = {Mathematische Annalen},
    VOLUME = {296},
      YEAR = {1993},
    NUMBER = {4},
     PAGES = {625--635},
       DOI = {10.1007/BF01445125},
}

@book {C97,
    AUTHOR = {Cassels, J. W. S.},
     TITLE = {An Introduction to the Geometry of Numbers},
    SERIES = {Classics in Mathematics},
      NOTE = {Corrected reprint of the 1971 edition},
 PUBLISHER = {Springer-Verlag, Berlin},
      YEAR = {1997},
     PAGES = {viii+344},
}

@article {DGS25,
    AUTHOR = {J. Dick and T. Goda and K. Suzuki},
     TITLE = {On the quasi-uniformity properties of quasi-{M}onte {C}arlo point sets and sequences -- {P}art~{II}: {D}igital nets and sequences},
   JOURNAL = {arXiv:2501.18226},
   URL = {https://arxiv.org/abs/2501.18226},
   YEAR = {2025}
}

@article {DGLPS25,
    AUTHOR = {J. Dick and T. Goda and G. Larcher and F. Pillichshammer and K. Suzuki},
     TITLE = {On the quasi-uniformity properties of quasi-{M}onte {C}arlo point sets and sequences -- {P}art~{I}: {L}attices and {K}ronecker sequences},
   JOURNAL = {arXiv:2502.06202},
   URL = {https://arxiv.org/abs/2502.06202},
   YEAR = {2025}
}

@book {DKP22,
    AUTHOR = {Dick, J. and Kritzer, P. and Pillichshammer, F.},
     TITLE = {Lattice Rules---Numerical Integration, Approximation, and
              Discrepancy},
    SERIES = {Springer Series in Computational Mathematics},
    VOLUME = {58},
 PUBLISHER = {Springer, Cham},
      YEAR = {2022},
     PAGES = {xvi+580},
       DOI = {10.1007/978-3-031-09951-9},
}

@article {DKS13,
    AUTHOR = {Dick, J. and Kuo, F.~Y. and Sloan, I.~H.},
     TITLE = {High-dimensional integration: the quasi-{M}onte {C}arlo way},
   JOURNAL = {Acta Numer.},
  FJOURNAL = {Acta Numerica},
    VOLUME = {22},
      YEAR = {2013},
     PAGES = {133--288},
       DOI = {10.1017/S0962492913000044},
}

@article {G24a,
    AUTHOR = {Goda, T.},
     TITLE = {The {S}obol' sequence is not quasi-uniform in dimension 2},
   JOURNAL = {Proc. Amer. Math. Soc.},
  FJOURNAL = {Proceedings of the American Mathematical Society},
    VOLUME = {152},
      YEAR = {2024},
    NUMBER = {8},
     PAGES = {3209--3213},
}

@article {G24b,
    AUTHOR = {Goda, T.},
     TITLE = {One-dimensional quasi-uniform {K}ronecker sequences},
   JOURNAL = {Arch. Math.},
  FJOURNAL = {Archiv der Mathematik},
    VOLUME = {123},
      YEAR = {2024},
    NUMBER = {5},
     PAGES = {499--505},
}

@article {GHS25,
    AUTHOR = {T. Goda and R. Hofer and K. Suzuki},
     TITLE = {{Disproving the quasi-uniformity of the Halton sequences and some of Halton-type sequences}},
   JOURNAL = {arXiv:2511.21153},
   URL = {https://arxiv.org/abs/2511.21153},
   YEAR = {2025}
}

@book {KN74,
    AUTHOR = {Kuipers, L. and Niederreiter, H.},
     TITLE = {Uniform Distribution of Sequences},
    SERIES = {Pure and Applied Mathematics},
 PUBLISHER = {Wiley-Interscience, New
              York-London-Sydney},
      YEAR = {1974},
     PAGES = {xiv+390},
}

@article {L88,
    AUTHOR = {Larcher, G.},
     TITLE = {On the distribution of {$s$}-dimensional {K}ronecker-sequences},
   JOURNAL = {Acta Arith.},
  FJOURNAL = {Polska Akademia Nauk. Instytut Matematyczny. Acta Arithmetica},
    VOLUME = {51},
      YEAR = {1988},
    NUMBER = {4},
     PAGES = {335--347},
       DOI = {10.4064/aa-51-4-335-347},
}

@article {MBC79,
    AUTHOR = {McKay, M.~D. and Beckman, R.~J. and Conover, W.~J.},
     TITLE = {A comparison of three methods for selecting values of input
              variables in the analysis of output from a computer code},
   JOURNAL = {Technometrics},
  FJOURNAL = {Technometrics. A Journal of Statistics for the Physical,
              Chemical and Engineering Sciences},
    VOLUME = {21},
      YEAR = {1979},
    NUMBER = {2},
     PAGES = {239--245},
       DOI = {10.2307/1268522},
}

@book {N92,
    AUTHOR = {Niederreiter, H.},
     TITLE = {Random Number Generation and Quasi-{M}onte {C}arlo methods},
    SERIES = {CBMS-NSF Regional Conference Series in Applied Mathematics},
    VOLUME = {63},
 PUBLISHER = {Society for Industrial and Applied Mathematics (SIAM),
              Philadelphia, PA},
      YEAR = {1992},
}

@article {LLL82,
    AUTHOR = {Lenstra, A. K. and Lenstra, Jr., H. W. and Lov\'asz, L.},
     TITLE = {Factoring polynomials with rational coefficients},
   JOURNAL = {Math. Ann.},
  FJOURNAL = {Mathematische Annalen},
    VOLUME = {261},
      YEAR = {1982},
    NUMBER = {4},
     PAGES = {515--534},
       DOI = {10.1007/BF01457454},
}

@article {NS09,
    AUTHOR = {Nguyen, P.~Q. and Stehl\'e, D.},
     TITLE = {An {LLL} algorithm with quadratic complexity},
   JOURNAL = {SIAM J. Comput.},
  FJOURNAL = {SIAM Journal on Computing},
    VOLUME = {39},
      YEAR = {2009},
    NUMBER = {3},
     PAGES = {874--903},
       DOI = {10.1137/070705702},
}

@book {NV10,
     TITLE = {The {LLL} Algorithm},
    SERIES = {Information Security and Cryptography},
    EDITOR = {Nguyen, P.~Q. and Vall\'ee, B.},
 PUBLISHER = {Springer-Verlag, Berlin},
      YEAR = {2010},
     PAGES = {xiv+496},
       DOI = {10.1007/978-3-642-02295-1},
}

@article {HE06,
    AUTHOR = {Hechenleitner, B. and Entacher, K.},
     TITLE = {A parallel search for good lattice points using {LLL}-spectral
              tests},
   JOURNAL = {J. Comput. Appl. Math.},
  FJOURNAL = {Journal of Computational and Applied Mathematics},
    VOLUME = {189},
      YEAR = {2006},
    NUMBER = {1-2},
     PAGES = {424--441},
       DOI = {10.1016/j.cam.2005.03.058},
}

@article {PZ23,
    AUTHOR = {Pronzato, L. and Zhigljavsky, A.},
     TITLE = {Quasi-uniform designs with optimal and near-optimal uniformity
              constant},
   JOURNAL = {J. Approx. Theory},
  FJOURNAL = {Journal of Approximation Theory},
    VOLUME = {294},
      YEAR = {2023},
     PAGES = {Paper No. 105931, 14},
}

@book {S80,
    AUTHOR = {Schmidt, W.~M.},
     TITLE = {Diophantine Approximation},
    SERIES = {Lecture Notes in Mathematics},
    VOLUME = {785},
 PUBLISHER = {Springer, Berlin},
      YEAR = {1980},
     PAGES = {x+299},
}

@book {S89,
    AUTHOR = {Siegel, C.~L.},
     TITLE = {Lectures on the Geometry of Numbers},
 PUBLISHER = {Springer-Verlag, Berlin},
      YEAR = {1989},
     PAGES = {x+160},
       DOI = {10.1007/978-3-662-08287-4},
}

@book {MG02,
    AUTHOR = {Micciancio, D. and Goldwasser, S.},
     TITLE = {Complexity of Lattice Problems},
    SERIES = {The Kluwer International Series in Engineering and Computer
              Science},
    VOLUME = {671},
 PUBLISHER = {Kluwer Academic Publishers, Boston, MA},
      YEAR = {2002},
     PAGES = {x+220},
       DOI = {10.1007/978-1-4615-0897-7},
}

@book {SJ94,
    AUTHOR = {Sloan, I. H. and Joe, S.},
     TITLE = {Lattice Methods for Multiple Integration},
    SERIES = {Oxford Science Publications},
 PUBLISHER = {The Clarendon Press, Oxford University Press, New York},
      YEAR = {1994},
     PAGES = {xii+239},
}

@article {Sch95,
    AUTHOR = {Schaback, R.},
     TITLE = {Error estimates and condition numbers for radial basis function interpolation},
   JOURNAL = {Adv. Comput. Math.},
  FJOURNAL = {Advances in Computational Mathematics},
    VOLUME = {3},
      YEAR = {1995},
    NUMBER = {3},
     PAGES = {251--264},
}

@article {SW06,
    AUTHOR = {Schaback, R. and Wendland, H.},
     TITLE = {Kernel techniques: {F}rom machine learning to meshless methods},
   JOURNAL = {Acta Numer.},
  FJOURNAL = {Acta Numerica},
    VOLUME = {15},
      YEAR = {2006},
     PAGES = {543--639},
}

@book {W05,
    AUTHOR = {Wendland, H.},
     TITLE = {Scattered Data Approximation},
    SERIES = {Cambridge Monographs on Applied and Computational Mathematics},
    VOLUME = {17},
 PUBLISHER = {Cambridge University Press, Cambridge},
      YEAR = {2005},
}

@article {T20,
    AUTHOR = {Teckentrup, A.~L.},
     TITLE = {Convergence of {G}aussian process regression with estimated
              hyper-parameters and applications in {B}ayesian inverse
              problems},
   JOURNAL = {SIAM/ASA J. Uncertain. Quantif.},
  FJOURNAL = {SIAM/ASA Journal on Uncertainty Quantification},
    VOLUME = {8},
      YEAR = {2020},
    NUMBER = {4},
     PAGES = {1310--1337},
}

@article {WSH21,
    AUTHOR = {Wenzel, T. and Santin, G. and Haasdonk, B.},
     TITLE = {A novel class of stabilized greedy kernel approximation algorithms: {C}onvergence, stability and uniform point distribution},
   JOURNAL = {J. Approx. Theory},
  FJOURNAL = {Journal of Approximation Theory},
    VOLUME = {262}, 
      YEAR = {2021},
     PAGES = {Paper No. 105508, 30},
       DOI = {10.1016/j.jat.2020.105508},
}

@article {WBG21,
    AUTHOR = {Wynne, G. and Briol, F.-X. and Girolami,
              M.},
     TITLE = {Convergence guarantees for {G}aussian process means with
              misspecified likelihoods and smoothness},
   JOURNAL = {J. Mach. Learn. Res.},
  FJOURNAL = {Journal of Machine Learning Research (JMLR)},
    VOLUME = {22},
      YEAR = {2021},
     PAGES = {Paper No. 123, 40},
}

@book {SWN03,
    AUTHOR = {Santner, T.~J. and Williams, B.~J. and Notz, W.~I.},
     TITLE = {The Design and Analysis of Computer Experiments},
    SERIES = {Springer Series in Statistics},
 PUBLISHER = {Springer-Verlag, New York},
      YEAR = {2003},
     PAGES = {xii+283},
}

@book {FLS06,
    AUTHOR = {Fang, K.-T. and Li, R. and Sudjianto, A.},
     TITLE = {Design and Modeling for Computer Experiments},
    SERIES = {Chapman \& Hall/CRC Computer Science and Data Analysis Series},
 PUBLISHER = {Chapman \& Hall/CRC, Boca Raton, FL},
      YEAR = {2006},
     PAGES = {xii+290},
}

@article {PM12,
    AUTHOR = {Pronzato, L. and M\"uller, W.~G.},
     TITLE = {Design of computer experiments: space filling and beyond},
   JOURNAL = {Stat. Comput.},
  FJOURNAL = {Statistics and Computing},
    VOLUME = {22},
      YEAR = {2012},
    NUMBER = {3},
     PAGES = {681--701},
}

@article {JMY90,
    AUTHOR = {Johnson, M. E. and Moore, L. M. and Ylvisaker, D.},
     TITLE = {Minimax and maximin distance designs},
   JOURNAL = {J. Statist. Plann. Inference},
  FJOURNAL = {Journal of Statistical Planning and Inference},
    VOLUME = {26},
      YEAR = {1990},
    NUMBER = {2},
     PAGES = {131--148},
       DOI = {10.1016/0378-3758(90)90122-B},
}

@article {P17,
    AUTHOR = {Pronzato, L.},
     TITLE = {Minimax and maximin space-filling designs: some properties and
              methods for construction},
   JOURNAL = {J. SFdS},
  FJOURNAL = {Journal de la SFdS. Journal de la Soc\'iet\'e{} Fran\c caise
              de Statistique},
    VOLUME = {158},
      YEAR = {2017},
    NUMBER = {1},
     PAGES = {7--36},
}

@article {HRDH11,
    AUTHOR = {Husslage, B. G. M. and Rennen, G. and van Dam, E. R. and Den Hertog, D.},
     TITLE = {Space-filling {L}atin hypercube designs for computer experiments},
   JOURNAL = {Optim. Eng.},
  FJOURNAL = {Optimization and Engineering },
    VOLUME = {12},
      YEAR = {2011},
    NUMBER = {1},
     PAGES = {610--633},
}

@article {J16,
    AUTHOR = {Joseph, V.~R.},
     TITLE = {{Space-filling designs for computer experiments: A review}},
   JOURNAL = {Qual. Eng.},
  FJOURNAL = {Quality Engineering},
    VOLUME = {28},
      YEAR = {2016},
    NUMBER = {1},
     PAGES = {28--35},
       DOI = {10.1080/08982112.2015.1100447},
}

@article {LE88,
    AUTHOR = {L'Ecuyer, P.},
     TITLE = {Efficient and portable combined random number generators},
   JOURNAL = {Comm. ACM},
  FJOURNAL = {Communications of the Association for Computing Machinery},
    VOLUME = {31},
      YEAR = {1988},
    NUMBER = {6},
     PAGES = {742--749, 774},
       DOI = {10.1145/62959.62969},
}

@article {LE99,
    AUTHOR = {L'Ecuyer, P.},
     TITLE = {Tables of linear congruential generators of different sizes
              and good lattice structure},
   JOURNAL = {Math. Comp.},
  FJOURNAL = {Mathematics of Computation},
    VOLUME = {68},
      YEAR = {1999},
    NUMBER = {225},
     PAGES = {249--260},
       DOI = {10.1090/S0025-5718-99-00996-5},
}

@article {F82,
    AUTHOR = {Franke, R.},
     TITLE = {Smooth interpolation of scattered data by local thin plate
              splines},
   JOURNAL = {Comput. Math. Appl.},
  FJOURNAL = {Computers \& Mathematics with Applications. An International
              Journal},
    VOLUME = {8},
      YEAR = {1982},
    NUMBER = {4},
     PAGES = {273--281},
       DOI = {10.1016/0898-1221(82)90009-8},
}

@inproceedings{G84,
author = {Genz, A.},
title = {Testing multidimensional integration routines},
year = {1984},
publisher = {Elsevier North-Holland, Inc.},
address = {USA},
booktitle = {Proc. of International Conference on Tools, Methods and Languages for Scientific and Engineering Computation},
pages = {81–94},
numpages = {14},
location = {Paris, France}
}

@book {F07,
    AUTHOR = {Fasshauer, G.~E.},
     TITLE = {Meshfree Approximation Methods with {MATLAB}},
    SERIES = {Interdisciplinary Mathematical Sciences},
    VOLUME = {6},
 PUBLISHER = {World Scientific Publishing Co. Pte. Ltd., Hackensack, NJ},
      YEAR = {2007},
     PAGES = {xviii+500},
       DOI = {10.1142/6437},
}

@article{KIWINTSG25,
     TITLE = {Space-filling {L}atin hypercube design for efficient {B}ayesian optimization With application to semiconductor development},
    AUTHOR = {Kinoshita, S. and Inoue, Y. and Watanabe, T. and Ikeda, K. and Nishio, S. and Teruya, A. and Sakai, N. and Goda, T.},
   JOURNAL = {IEEE Trans. Semicond. Manuf.},
  FJOURNAL = {IEEE Transactions on Semiconductor Manufacturing},
    VOLUME = {38},
      YEAR = {2025},
    NUMBER = {3},
     PAGES = {446--452},
       DOI = {10.1109/TSM.2025.3574791},
}

@article {B88,
    AUTHOR = {Beck, J.},
     TITLE = {On the discrepancy of convex plane sets},
   JOURNAL = {Monatsh. Math.},
  FJOURNAL = {Monatshefte f\"ur Mathematik},
    VOLUME = {105},
      YEAR = {1988},
    NUMBER = {2},
     PAGES = {91--106},
       DOI = {10.1007/BF01501162},
}

@article {S77,
    AUTHOR = {Stute, Winfried},
     TITLE = {Convergence rates for the isotrope discrepancy},
   JOURNAL = {Ann. Probability},
  FJOURNAL = {The Annals of Probability},
    VOLUME = {5},
      YEAR = {1977},
    NUMBER = {5},
     PAGES = {707--723},
       DOI = {10.1214/aop/1176995714},
}

@article {S75,
    AUTHOR = {Schmidt, W.~M.},
     TITLE = {Irregularities of distribution. {IX}},
   JOURNAL = {Acta Arith.},
  FJOURNAL = {Acta Arithmetica},
    VOLUME = {27},
      YEAR = {1975},
     PAGES = {385--396},
       DOI = {10.4064/aa-27-1-385-396},
}

@article{SHCRRJ26, 
    author = {Sorokin, A.~G. and Hickernell, F.~J. and Choi, S.-C.~T. and Rathinavel, J. and Robbe, P. and Jain, A.}, 
     title = {{QMCPy}: {A} {P}ython framework for ({Q}uasi-){M}onte {C}arlo algorithms}, 
   journal = {J. Open Source Softw.},
  fjournal = {Journal of Open Source Software},
      year = {2026}, 
    volume = {11}, 
    number = {117}, 
     pages = {9705}, 
       doi = {10.21105/joss.09705}, 
}

@article {LF03,
    AUTHOR = {Larsson, E. and Fornberg, B.},
     TITLE = {A numerical study of some radial basis function based solution
              methods for elliptic {PDE}s},
   JOURNAL = {Comput. Math. Appl.},
  FJOURNAL = {Computers \& Mathematics with Applications. An International
              Journal},
    VOLUME = {46},
      YEAR = {2003},
    NUMBER = {5-6},
     PAGES = {891--902},
       DOI = {10.1016/S0898-1221(03)90151-9},
}

@book {RW06,
    AUTHOR = {Rasmussen, C.~E. and Williams, C.~K.~I.},
     TITLE = {Gaussian Processes for Machine Learning},
    SERIES = {Adaptive Computation and Machine Learning},
 PUBLISHER = {MIT Press, Cambridge, MA},
      YEAR = {2006},
     PAGES = {xviii+248},
}

\end{document}